\newif\ifARXIV
\ARXIVtrue

\ifARXIV
\documentclass[runningheads]{llncs}
\else
\documentclass{llncs}
\fi

\usepackage{amssymb,amsmath}
\newtheorem{observation}{Observation}
\usepackage{xspace}
\usepackage{xcolor}
\usepackage{paralist}

\ifARXIV
\usepackage{geometry} 
\fi
\usepackage{latexsym}
\usepackage{hyperref}
\hypersetup{pdfborder={0 0 0.4}}
\usepackage{graphicx}
\usepackage{cite}
\usepackage{lineno}
\usepackage{subfigure}

\ifARXIV
\fi

\spnewtheorem*{sketch}{Proof sketch}{\itshape}{\rmfamily}

\title{Clustered Planarity = Flat Clustered Planarity\thanks{\ack}}

\newcommand{\ack}{This research was partially supported by MIUR
project ``MODE -- MOrphing graph Drawings Efficiently'',
prot.~20157EFM5C\_001.
}

\authorrunning{P. F. Cortese and M. Patrignani}

\author{Pier Francesco Cortese \and Maurizio~Patrignani}
  
\institute{Roma Tre University, Rome, Italy\\
\email{pierfrancesco@pfcortese.it} \\
\email{maurizio.patrignani@uniroma3.it} 
}

\graphicspath{{figures/}}

\definecolor{blue}{rgb}{0.274,0.392,0.666}
\definecolor{ourred}{rgb}{1,0.3,0.3}
\definecolor{ourgreen}{rgb}{0,0.588,0.509}

\newcommand{\remove}[1]{}

\newcommand{\Hconnected}{$\mathcal{H}$\emph{-conn}\xspace}
\newcommand{\Hnotroot}{$\mathcal{H}$\emph{-not-root}\xspace}

\let\doendproof\endproof
\renewcommand\endproof{\hfill $\qed$\doendproof}

\let\doendsketch\endsketch
\renewcommand\endsketch{\hfill $\qed$\doendsketch}

\begin{document}

\maketitle
\begin{abstract}
 The complexity of deciding whether a clustered graph admits a clustered planar drawing is a long-standing open problem in the graph drawing research area. Several research efforts focus on a restricted version of this problem where the hierarchy of the clusters is `flat', i.e., no cluster different from the root contains other clusters. We prove that this restricted problem, that we call \textsc{Flat Clustered Planarity}, retains the same complexity of the general \textsc{Clustered Planarity} problem, where the clusters are allowed to form arbitrary hierarchies.
 We strengthen this result by showing that \textsc{Flat Clustered Planarity} is polynomial-time equivalent to \textsc{Independent Flat Clustered Planarity}, where each cluster induces an independent set. We discuss the consequences of these results.
\end{abstract}

\section{Introduction}

A clustered graph (c-graph) is a planar graph with a recursive hierarchy defined on its vertices. A clustered planar (c-planar) drawing of a c-graph is a planar drawing of the underlying graph where: (i) each cluster is represented by a simple closed region of the plane containing only the vertices of the corresponding cluster, (ii) cluster borders never intersect, and (iii) any edge and any cluster border intersect at most once (more formal definitions are given in Section~\ref{se:preliminaries}).
The complexity of deciding whether a c-graph admits a c-planar drawing is still an open problem after more than 20 years of intense research~\cite{afp-sctgc-09,br-npcpc-16,cw-cccg-06,cd-cp-05,cdfpp-cccg-j-08,d-ltarc-98,ef-mvcg-96,efn-dcgog-99,eh-ncguf-00,fce-pcg-95,f-cgd-14,fkmp-cptr-15,gls-cpecg-05,gjlmpw-acptc-02,gms-pehtt-14,jstv-cpcfoe-09,jjkl-cpecg-08,s-ttpht-13}.

If we had an efficient c-planarity testing and embedding algorithm we could produce straight-line drawings of clustered trees~\cite{j-ddf-hdct-09} and straight-line drawings~\cite{afk-srdcg-11,efln-sldah-06} and orthogonal drawings~\cite{ddm-pcg-02} of c-planar c-graphs with rectangular regions for the clusters.

In order to shed light on the complexity of \textsc{Clustered Planarity}, this problem has been compared with other problems whose complexity is likewise challenging. This line of investigation was opened by Marcus Schaefer's polynomial-time reduction of \textsc{Clustered Planarity} to \textsc{SEFE}~\cite{s-ttpht-13}. \textsc{Simultaneous Embedding with Fixed Edges} (\textsc{SEFE}) takes as input two planar graphs $G_1=(V,E_1)$ and $G_2=(V, E_2)$ and asks whether a planar drawing~$\Gamma_1(G_1)$ and a planar drawing~$\Gamma_2(G_2)$ exist such that: (i) each vertex $v \in V$ is mapped to the same point in $\Gamma_1$ and in $\Gamma_2$ and (ii) every edge $e \in E_1 \cap E_2$ is mapped to the same Jordan curve in $\Gamma_1$ and in $\Gamma_2$. 

However, the polynomial-time equivalence of the two problems is open and the reverse reduction of \textsc{SEFE} to \textsc{Clustered Planarity} is known only for the case when the intersection graph $G_\cap=(V,E_1 \cap E_2)$ of the instance of \textsc{SEFE} is connected~\cite{ad-s=c-thecp-16}. Also in this special case, the complexity of the problem is unknown, with the exception of the case when $G_\cap$ is a star, which produces a c-graph with only two clusters, a known polynomial case for \textsc{Clustered Planarity}~\cite{adfpr-tsetg-12,hn-tpbecp-09}.

Since the general \textsc{Clustered Planarity} problem appears to be elusive, several authors focused on a restricted version of it where the hierarchy of the clusters is `flat', i.e., only the root cluster contains other clusters and it does not directly contain vertices of the underlying graph \cite{aft-rweg-18,ad-cpwp-16,addf-spt-13,addf-sptepg-16,addfpr-rccp-14,addfr-tibp-15,cdfk-atcefcg-14,cdpp-cccc-05,cdpp-ecpg-09,degg-sfefcg-conf-17,df-ectefcgsf-j-09,f-mdcg-j14,f-cpecc-17,f-egeg-17,fk-htamg-18,hn-tpbecp-09,jkkpsv-cpscc-09}. This restricted problem, that we call \textsc{Flat Clustered Planarity}, is expressive enough to be useful in several applicative domains, as for example in computer networks where routers are grouped into Autonomous Systems~\cite{cdds-rvet-18}, or social networks where people are grouped into communities~\cite{bbdlpp-valgu-11,ddlp-vaotm-10}, or software diagrams where classes are grouped into packages~\cite{s-vc-96}.
Also, several hybrid representations have been proposed for the visual analysis of (not necessarily planar) flat clustered graphs, such as mixed matrix and node-link representations~\cite{bbdlpp-valgu-11,ddfp-cnrcg-16,ddfp-cnrcg-jgaa-17,dlpt-nptsc-17,hfm-dhvsn-07}, mixed intersection and node-link representations~\cite{addfpr-ilrg-17}, and mixed space-filling and node-link representations~\cite{aky-vlgcfvt-04,imms-hsffd-09,zmc-ehctn-05}.

Unfortunately, the complexity of \textsc{Flat Clustered Planarity} is open as the complexity of the general problem. The authors of~\cite{br-npcpc-16}, after recasting \textsc{Flat Clustered Planarity} as an embedding problem on planar multi-graphs, conclude that we are still far away from solving it. The authors of~\cite{ad-s=c-thecp-16} wonder whether \textsc{Flat Clustered Planarity}
retains the same complexity of \textsc{Clustered Planarity}. In this paper we answer this question in the affirmative.
Obviously, a reduction of \textsc{Flat Clustered Planarity} to \textsc{Clustered Planarity} is trivial, since the instances of \textsc{Flat Clustered Planarity} are simply a subset of those of \textsc{Clustered Planarity}. The reverse reduction is the subject of Section~\ref{se:reduction}, that proves the following theorem.

\begin{theorem}\label{th:main}
There exists a quadratic-time transformation that maps an instance of \textsc{Clustered Planarity} to an equivalent instance of \textsc{Flat Clustered Planarity}.
\end{theorem}

With very similar techniques we are able to prove also a stronger result.  

\begin{theorem}\label{th:main2}
There exists a linear-time transformation that maps an instance of \textsc{Flat Clustered Planarity} to an equivalent instance of \textsc{Independent Flat Clustered Planarity}.
\end{theorem}

Here, by \textsc{Independent Flat Clustered Planarity} we mean the restriction of \textsc{Flat Clustered Planarity} to instances where each non-root cluster induces an independent set. 

The paper is structured as follows. Section~\ref{se:preliminaries} contains basic definitions. Section~\ref{se:reduction} contains the proof of Theorem~\ref{th:main} under some simplifying hypotheses 
\ifARXIV
(which are removed in~\ref{ap:lemma}).
\else 
(which are removed in~\cite{cp-cpfcp-tr-18}).
\fi
Some immediate consequences of Theorem~\ref{th:main} are discussed in Section~\ref{se:discussion}. The proof of Theorem~\ref{th:main2} and some remarks about it are in Sections~\ref{se:reduction2} and~\ref{se:discussion2}, respectively. Conclusions and open problems are in Section~\ref{se:conclusions}. 
\ifARXIV
For space reasons some proofs are moved to the appendix.
\else
For space reasons some proofs are sketched or, when trivial, omitted.
\fi

\section{Preliminaries}\label{se:preliminaries}

Let $T$ be a rooted tree.
We denote by $r(T)$ the root of $T$ and by $T[\mu]$ the subtree of $T$ rooted at one of its nodes~$\mu$. 
The \emph{depth} of a node $\mu$ of $T$ is the length (number of edges) of the path from $r(T)$ to~$\mu$. The \emph{height} $h(T)$ of a tree $T$ is the maximum depth of its nodes. 

The nodes of a tree can be partitioned into \emph{leaves}, that do not have children, and \emph{internal nodes}. In turn, the internal nodes can be partitioned into two sets: \emph{lower nodes}, whose children are all leaves, and \emph{higher nodes}, that have at least one internal-node child. 
We say that a node is \emph{homogeneous} if its children are either all leaves or all internal nodes. A tree is  \emph{homogeneous} if all its nodes are homogeneous.
We say that a tree is \emph{flat} if all its leaves have depth~$2$. A flat tree is homogeneous. 
Figure~\ref{fi:trees} shows a non-homogeneous tree (Fig.~\ref{fi:non-homogeneous}), a homogeneous tree (Fig.~\ref{fi:homogeneous}), and a flat tree (Fig.~\ref{fi:flat}). 

\begin{figure}[htb]
\centering
\subfigure[]{\includegraphics[page=1,height=0.10\textwidth]{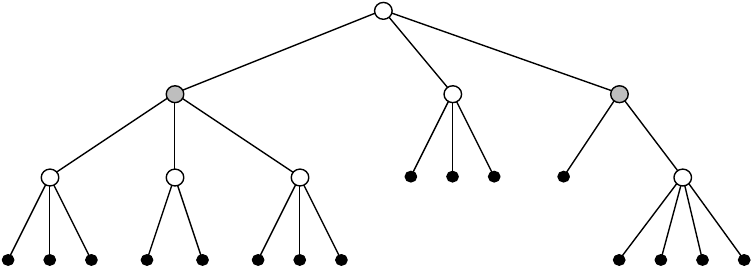}\label{fi:non-homogeneous}}
\hfil
\subfigure[]{\includegraphics[page=1,height=0.10\textwidth]{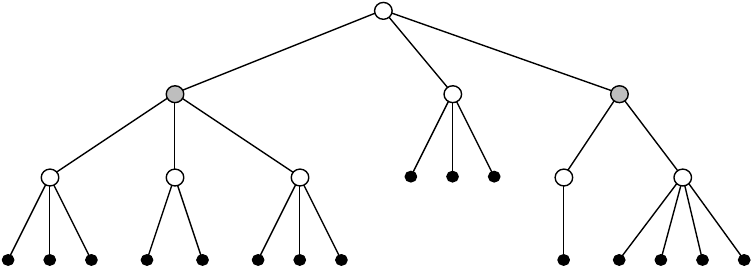}\label{fi:homogeneous}}
\hfil
\subfigure[]{\includegraphics[page=1,height=0.08\textwidth]{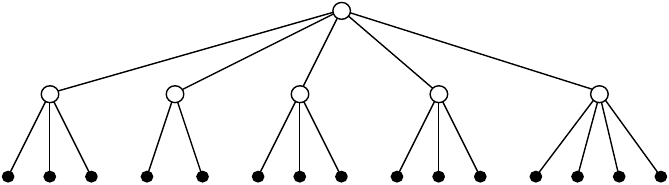}\label{fi:flat}}
\hfil
\caption{\subref{fi:non-homogeneous} A tree that is not homogeneous. \subref{fi:homogeneous} A homogeneous tree. \subref{fi:flat} A flat tree. 
}\label{fi:trees}
\end{figure}

We also need a special notion of size: the \emph{size} of a tree $T$, denoted by $\mathcal{S}(T)$, is the number of higher nodes of $T$ different from the root of $T$. 
Observe that a homogeneous tree $T$ is flat if and only if $\mathcal{S}(T)=0$.
For example, the sizes of the trees represented in Figs.~\ref{fi:non-homogeneous}, \ref{fi:homogeneous}, and~\ref{fi:flat} are $2$, $2$, and $0$, respectively (filled gray nodes in Fig.~\ref{fi:trees}). 
\ifARXIV
The proof of the following lemma can be found in~\ref{ap:lemmata}.
\else
The proof of the following lemma is trivial.
\fi

\begin{lemma}\label{le:homogeneous-subtree}
A homogeneous tree $T$ of height $h(T) \geq 2$ and size $\mathcal{S}(T)>0$ contains at least one node $\mu^* \neq r(T)$ such that $T[\mu^*]$ is flat.
\end{lemma}

A \emph{graph} $G=(V,E)$ is a set $V$ of \emph{vertices} and a set $E$ of \emph{edges}, where each edge is an unordered pair of vertices. A \emph{drawing} $\Gamma(G)$ of $G$ is a mapping of its vertices to distinct points on the plane and of its edges to Jordan curves joining the incident vertices.
Drawing $\Gamma(G)$ is \emph{planar} if no two edges intersect except at common end-vertices. A graph is \emph{planar} if it admits a planar drawing.

A \emph{clustered graph} (or \emph{c-graph}) $C$ is a pair $(G,T)$ where $G=(V,E)$ is a planar graph, called the \emph{underlying graph} of $C$, and $T$, called the \emph{inclusion tree} of $C$, is a rooted tree such that the set of leaves of $T$ coincides with $V$. A \emph{cluster} $\mu$ is an internal node of $T$. When it is not ambiguous we also identify a cluster with the respective subset of the vertex set. 
An \emph{inter-cluster edge} of a cluster $\mu$ of $T$ is an edge of $G$ that has one end-vertex inside $\mu$ and the other end-vertex outside~$\mu$. 
An \emph{independent set} of vertices is a set of pairwise non-adjacent vertices. A cluster $\mu$ of $T$ is \emph{independent} if its vertices form an independent set. A c-graph is \emph{independent} if all its clusters, with the exception of the root, are independent clusters.
A cluster $\mu$ of $T$ is a \emph{lower cluster} (\emph{higher cluster}) of $C$ if $\mu$ is a lower node (higher node) of $T$. 

A c-graph is \emph{flat} if its inclusion tree is flat. The clusters of a flat c-graph are all lower clusters with the exception of the root cluster. A cluster is called \emph{singleton} if it contains a single cluster or a single vertex.  

A \emph{drawing} $\Gamma(C)$ of a c-graph $C(G,T)$ is a mapping of vertices and edges of $G$ to points and to Jordan curves joining their incident vertices, respectively, and of each internal node $\mu$ of $T$ to a simple closed region $R(\mu)$ containing exactly the vertices of $\mu$. 
Drawing $\Gamma(C)$ is \emph{c-planar} if: (i) curves representing edges of $G$ do not intersect except at common end-points; (ii) the boundaries of the regions representing clusters do not intersect; and (iii) each edge intersects the boundary of a region at most one time.
A c-graph is \emph{c-planar} if it admits a c-planar drawing.

Problem \textsc{Clustered Planarity} is the problem of deciding whether a c-graph is c-planar. Problem \textsc{Flat Clustered Planarity} is the restriction of \textsc{Clustered Planarity} to flat c-graphs. Problem \textsc{Independent Flat Clustered Planarity} is the restriction of \textsc{Clustered Planarity} to independent flat c-graphs.

%
\ifARXIV
The proof of the following lemmas can be found in~\ref{ap:lemmata}.
\else
The proof of the following lemmas can be found in~\cite{cp-cpfcp-tr-18}.
\fi


\remove{
\begin{lemma}\label{le:t-homogeneous}
\textsc{Clustered Planarity} can be reduced in linear time to the case when the inclusion tree is homogeneous.
\end{lemma}

\begin{lemma}\label{le:h-and-n} 
\textsc{Clustered Planarity} can be reduced in linear time to the case when the root of the inclusion tree $T$ has at least two children and $h(T) \leq n$, where $n$ is the number of vertices of the underlying graph.
\end{lemma}
}

\begin{lemma}\label{le:preconditions}
An instance $C(G,T)$ of \textsc{Clustered Planarity} with $n$ vertices and $c$ clusters can be reduced in time $O(n+c)$ to an equivalent instance such that:
\begin{inparaenum}[(1)]
\item\label{le:preconditions1} $T$ is homogeneous,
\item\label{le:preconditions2} $r(T)$ has at least two children, and 
\item\label{le:preconditions3} $h(T) \leq n-1$.
\end{inparaenum}
\end{lemma}

\section{Proof of Theorem~\ref{th:main}}\label{se:reduction}

We describe a polynomial-time reduction of \textsc{Clustered Planarity} to \textsc{Flat Clustered Planarity}.
Let $C(G,T)$ be a clustered graph, let $n$ be the number of vertices of~$G$, and let $c$ be the number of clusters of $C$. Due to Lemma~\ref{le:preconditions} we can achieve in $O(n+c)$ time that $T$ is homogeneous and $\mathcal{S}(T) \in O(n)$. We reduce~$C$ to an equivalent instance $C_f(G_f,T_f)$ where $T_f$ is flat. The reduction consists of a sequence of transformations of $C=C_0$ into $C_1$, $C_2$, \dots, $C_{\mathcal{S}(T)} = C_f$, where each $C_i(G_i,T_i)$, $i=0, 1, \dots, {\mathcal{S}(T)}$, has an homogeneous inclusion tree $T_i$ and each transformation takes $O(n)$ time. 

\begin{figure}[tb]
\centering
\subfigure[]{\includegraphics[page=1,height=0.28\textwidth]{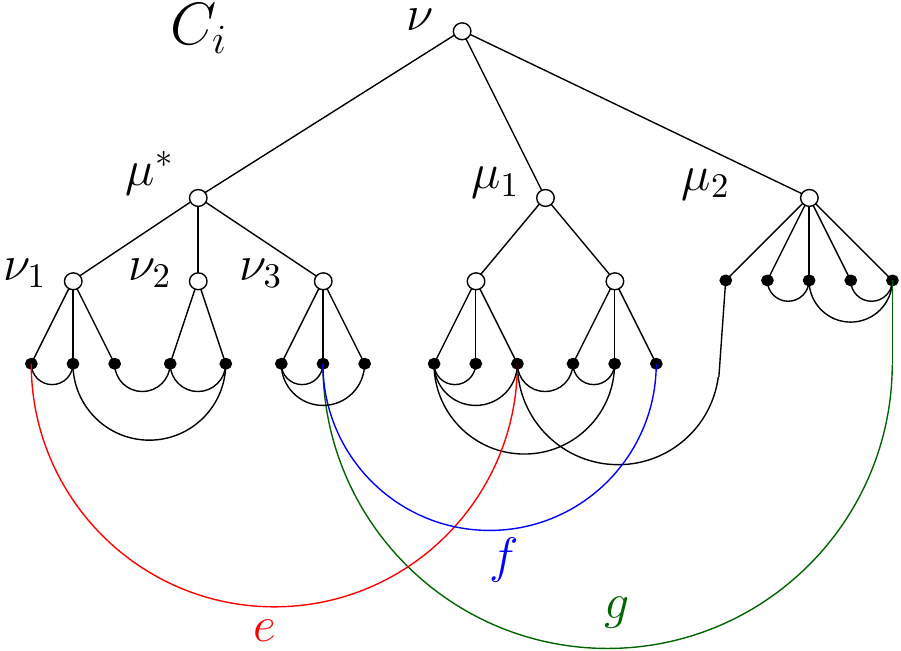}\label{fi:c-graph-c-i}}
\hfil
\subfigure[]{\includegraphics[page=2,height=0.28\textwidth]{figures/construction.pdf}\label{fi:c-graph-c-i+1}}
\caption{\subref{fi:c-graph-c-i} A c-graph $C_i$. Inter-cluster edges of $\mu^*$ are colored red, green, and blue. \subref{fi:c-graph-c-i+1} The construction of $C_{i+1}$.}\label{fi:construction}
\end{figure}

Consider any $C_i(G_i,T_i)$, with $i=0,\dots,{\mathcal{S}(T)-1}$, where $T_i$ is a homogeneous, non-flat tree of height $h(T_i) \geq 2$ (refer to Fig.~\ref{fi:c-graph-c-i}). 
By Lemma~\ref{le:homogeneous-subtree}, $T_i$ has at least one node $\mu^* \neq r(T_i)$ such that $T_i[\mu^*]$ is flat. Since $\mu^* \neq r(T_i)$, node $\mu^*$ has a parent $\nu$. Also, denote by $\nu_1, \nu_2, \dots, \nu_h$ the children of $\mu^*$ and by $\mu_1, \mu_2, \dots, \mu_k$ the siblings of $\mu^*$ in~$T_i$.
We construct $C_{i+1}(G_{i+1},T_{i+1})$ as follows (refer to Fig.~\ref{fi:c-graph-c-i+1}). Graph $G_{i+1}$ is obtained from $G_i$ by introducing, for each inter-cluster edge $e=(u,v)$ of $\mu^*$, two new vertices $e_\chi$ and $e_\varphi$ and by replacing $e$ with a path $(u,e_\chi)(e_\chi,e_\varphi)(e_\varphi,v)$.
Tree $T_{i+1}$ is obtained from $T_i$ by removing node $\mu^*$, attaching its children $\nu_1, \nu_2, \dots, \nu_h$ directly to $\nu$ and adding to $\nu$ two new children $\chi$ and $\varphi$, where cluster $\chi$ (cluster $\varphi$, respectively) contains all vertices $e_{\chi}$ ($e_\varphi$, respectively) introduced when replacing each inter-cluster edge $e$ of $\mu^*$ with a path. 
\ifARXIV
The proof of the following lemmas can be found in~\ref{ap:lemma}.
\else
The following lemmas are trivial.
\fi

\begin{lemma}\label{le:homogeneous-preserved}
If $T_i$ is homogeneous then $T_{i+1}$ is homogeneous. 
\end{lemma}

\begin{lemma}\label{le:size-decreases}
We have that $\mathcal{S}(T_{i+1}) = \mathcal{S}(T_i) - 1$. 
\end{lemma}

\begin{lemma}\label{le:final-flat}
The c-graph $C_f=C_{\mathcal{S}(T)}$ is flat.
\end{lemma}

\begin{figure}[tb]
\centering
\subfigure[]{\includegraphics[page=1,height=0.28\textwidth]{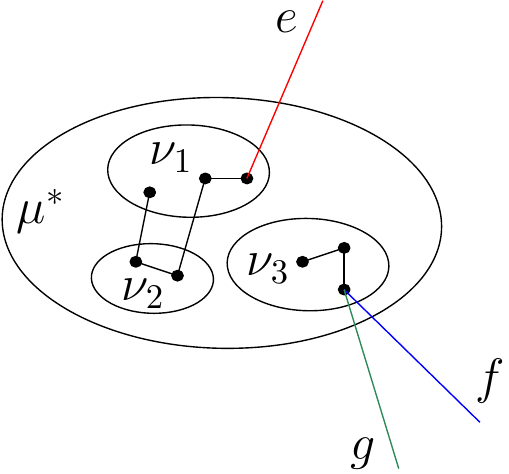}\label{fi:c2fc1}}
\hfil
\subfigure[]{\includegraphics[page=2,height=0.28\textwidth]{figures/c2fc.pdf}\label{fi:c2fc2}}
\caption{\subref{fi:c2fc1} A c-planar drawing $\Gamma(C_i)$ of c-graph $C_i$. \subref{fi:c2fc2} The construction of a c-planar drawing $\Gamma(C_{i+1})$.}\label{fi:c2fc}
\end{figure}

The proof of the following lemma is given here under two simplifying hypotheses 
\ifARXIV
(the proof of the general case can be found in \ref{ap:lemma}):
\else
(the proof of the general case can be found in \cite{cp-cpfcp-tr-18}):
\fi

\begin{quote}
\begin{enumerate}
\item[\Hconnected:] The underlying graph $G_i$ is connected 
\item[\Hnotroot:] Cluster $\nu$ is not the root of~$T$
\end{enumerate} 
\end{quote}

Observe that Hypothesis \Hconnected implies that also $G_{i+1}$ is connected.
Observe, also, that Hypothesis \Hnotroot and Property~\ref{le:preconditions2} of Lemma~\ref{le:preconditions} imply that there is at least one vertex of~$G_i$ that is not part of $\nu$ (this hypothesis is not satisfied, for example, by the c-graph depicted in Fig.~\ref{fi:c-graph-c-i}).

\begin{lemma}\label{le:main}
$C_i(G_i,T_i)$ is c-planar if and only if $C_{i+1}(G_{i+1},T_{i+1})$ is c-planar.
\end{lemma}
\begin{sketch}
The first direction of the proof is straightforward. Let $\Gamma(C_i)$ be a c-planar drawing of $C_i$ (refer to Fig.~\ref{fi:c2fc1}). We show how to construct a c-planar drawing of $C_{i+1}$ (refer to Fig.~\ref{fi:c2fc2}). Consider the region $R(\mu^*)$ that contains $R(\nu_i)$, with $i=1,\dots,h$. The boundary of $R(\mu^*)$ is crossed exactly once by each inter-cluster edge of $\mu^*$. Identify outside the boundary of $R(\mu^*)$ two arbitrarily thin regions $R(\chi)$ and $R(\varphi)$ that turn around $R(\mu^*)$ and that intersect exactly once all and only the inter-cluster edges of $\mu^*$. Insert into each inter-cluster edge $e$ of $\mu^*$ two vertices $e_\chi$ and $e_\varphi$, placing $e_\chi$ inside $R(\chi)$ and $e_\varphi$ inside $R(\varphi)$. By ignoring $R(\mu^*)$ you have a c-planar drawing $\Gamma(C_{i+1})$ of $C_{i+1}$.

\begin{figure}[tb]
\centering
\includegraphics[page=1,width=0.70\textwidth]{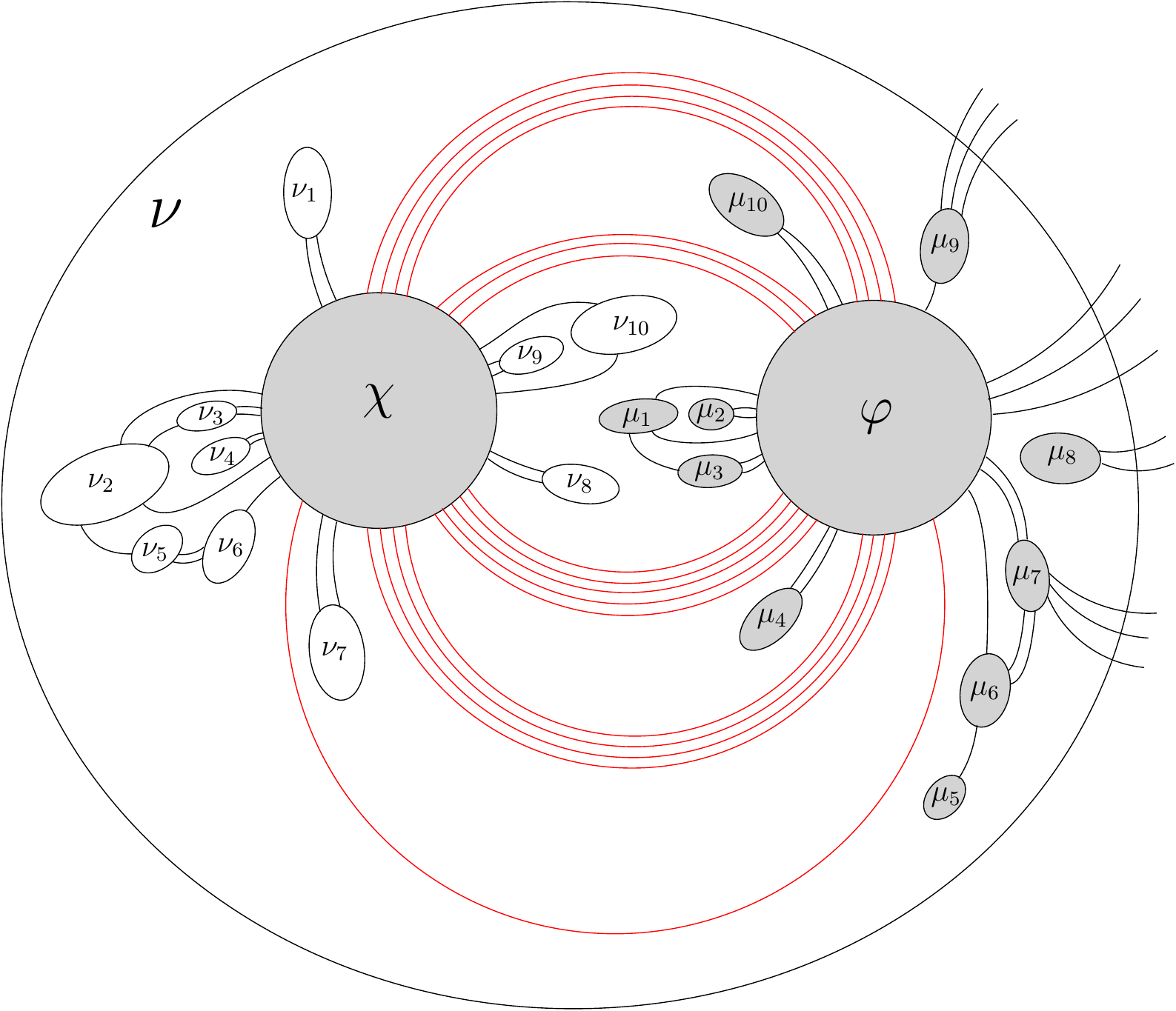}
\caption{A c-planar drawing of clusters $\nu$, $\chi$, and $\varphi$ in $\Gamma(C_{i+1})$.}\label{fi:fc2c1}
\end{figure}

Suppose now to have a c-planar drawing $\Gamma(C_{i+1})$ of $C_{i+1}$. We show how to construct a c-planar drawing $\Gamma(C_i)$ of $C_{i}$ under the Hypotheses \Hconnected and \Hnotroot. Consider the regions $R(\chi)$ and $R(\varphi)$ inside $R(\nu)$ (refer to Fig.~\ref{fi:fc2c1}). Regions $R(\chi)$ and $R(\varphi)$ are joined by the $p$ inter-cluster edges introduced when replacing each inter-cluster edge $e_i$ of $\mu^*$, where $i=1, \dots, p$, with a path (red edges of Fig.~\ref{fi:fc2c1}). Such inter-cluster edges of $\chi$ and $\varphi$ partition $R(\nu)$ into $p$ regions that have to host the remaining children of $\nu$ and the inter-cluster edges among them. In particular, $p-1$ of these regions are bounded by two inter-cluster edges and two portions of the boundaries of $R(\chi)$ and $R(\varphi)$. One of such regions, instead, is also externally bounded by the boundary of $R(\nu)$.   

Now consider the regions $R(\nu_i)$ corresponding to the children $\nu_i$ of $\nu$, with $i=1,\dots,h$, that were originally children of $\mu^*$. These regions (filled white in Fig.~\ref{fi:fc2c1}) may have inter-cluster edges among them and may be connected to $\chi$, but by construction cannot have inter-cluster edges connecting them to $\varphi$, or connecting them to the original children $\mu_i \neq \mu^*$ of $\nu$, or exiting the border of $R(\nu)$. In particular, due to Hypothesis \Hconnected, these regions must be directly or indirectly connected to~$\chi$. Finally, consider the regions $R(\mu_i)$ corresponding to the original children $\mu_i \neq \mu^*$ of $\nu$ (filled gray in Fig.~\ref{fi:fc2c1}). These regions may have inter-cluster edges among them, connecting them to $\varphi$, or connecting them to the rest of the graph outside $\nu$. In particular, due to Hypotheses \Hconnected and \Hnotroot, each $\mu_i$ (and also $\varphi$) must be directly or indirectly connected to the border of $R(\nu)$. It follows that the drawing in $\Gamma(C_{i+1})$ of the subgraph $G_{\mu^*}$ composed by the regions of $\chi, \nu_1, \nu_2, \dots, \nu_h$ and their inter-cluster edges cannot contain in one of its internal faces any other cluster of $\nu$. Hence, the sub-region $R(\mu^*)$ of $R(\nu)$ that is the union of $R(\chi)$ and the region enclosed by $G_{\mu^*}$ is a closed and simple region that only contains the regions $R(\nu_1)$, \dots $R(\nu_h)$ plus the region $R(\chi)$ and all the inter-cluster edges among them (see Fig.~\ref{fi:fc2c2}). By ignoring $R(\chi)$ and $R(\varphi)$ and by removing vertices $e_\chi$ and $e_\varphi$ and joining their incident edges we obtain a c-planar drawing $\Gamma(C_i)$.   
\end{sketch}


\begin{figure}[tb]
\centering
\includegraphics[page=2,width=0.70\textwidth]{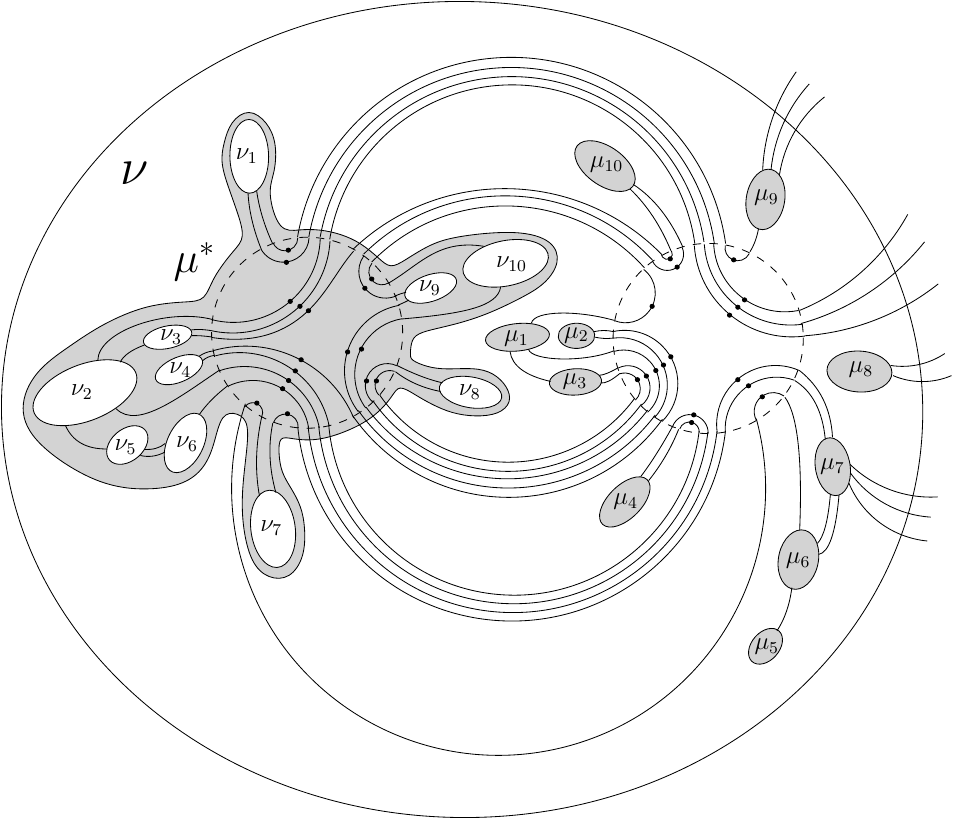}
\caption{The drawing of cluster $\mu$ in $\Gamma(C_i)$ corresponding to the drawing $\Gamma(C_{i+1})$ of Fig.~\ref{fi:fc2c1}.}\label{fi:fc2c2}
\end{figure}

The proof of Theorem~\ref{th:main} descends from Lemmas~\ref{le:final-flat} and~\ref{le:main} and from the consideration that each construction of $C_{i+1}$ from $C_i$ takes at most $O(n)$ time and, hence, the time needed to construct $C_f$ is $O(n^2)$. Due to the $O(n+c)$-time preprocessing (Lemma~\ref{le:preconditions}), the overall time complexity of the reduction is $O(n^2 + c)$.

\section{Remarks about Theorem~\ref{th:main}}\label{se:discussion}

In this section we discuss some consequences of Theorem~\ref{th:main} that descend from the properties of the reduction described in Section~\ref{se:reduction}. Such properties are summarized in the following lemma.

\begin{lemma}\label{le:properties}
Let $C(G,T)$ be an $n$-vertex clustered graph with $c$ clusters. The flat clustered graph $C_f(G_f,T_f)$ equivalent to $C$ built as described in the proof of Theorem~\ref{th:main} has the following properties:
\begin{enumerate}
\item\label{pr:subdivision} Graph $G_f$ is a subdivision of $G$
\item\label{pr:edges} Each edge of $G$ is replaced by a path of length at most $4h(T)-8$
\item\label{pr:vertices} The number of vertices of $G_f$ is $n_f \in O(n \cdot h(T))$
\item\label{pr:clusters} The number of clusters of $C_f$ is $c_f = c + \mathcal{S}(T)$
\end{enumerate}
\end{lemma}
\begin{proof}
Regarding Property~\ref{pr:subdivision}, observe that, for $i=1, \dots, \mathcal{S}(T)$, each $G_i$ is obtained from $G_{i-1}$ by replacing edges with paths. Hence $G_{\mathcal{S}(T)}=G_f$ is a subdivision of $G_0=G$. 
To prove Property~\ref{pr:edges} observe that each time an edge $e$ is subdivided, a pair of vertices $e_\chi$ and $e_\varphi$ is inserted and that edges are subdivided when the boundary of a higher cluster is removed. Edges that traverse more boundaries are those that link two vertices whose lowest common ancestor is the root of $T$. These edges traverse $2h(T)-4$ higher-cluster boundaries in~$C$. Hence, the number of vertices inserted into these edges is $4h(T)-8$. 
Property~\ref{pr:vertices} can be proved by considering that $G$ has $O(n)$ edges and each edge, by Property~\ref{pr:edges}, is replaced by a path of length at most $O(h(T))$.
Finally, Property~\ref{pr:clusters} descends from the fact that at each step $C_{i+1}$ has exactly one cluster more than $C_{i}$, since new clusters $\chi$ and $\varphi$ are inserted but cluster $\mu^*$ is removed.
\end{proof}

An immediate consequence of Property~\ref{pr:subdivision} of Lemma~\ref{le:properties} is that the number of faces of $G_f$ is equal to the number of faces of $G$. Also, if $G$ is connected, biconnected, or a subdivision of a triconnected graph, $G_f$ is also connected, biconnected, or a subdivision of a triconnected graph, respectively. If $G$ is a cycle or a tree, $G_f$ is also a cycle or a tree, respectively. Hence, the complexity of \textsc{Clustered Planarity} restricted to these kinds of graphs can be related to the complexity of \textsc{Flat Clustered Planarity} restricted to the same kinds of graphs.
Further, since a subdivision preserves the embedding of the original graph, the problem of deciding whether a c-graph $C(G,T)$ admits a c-planar drawing where $G$ has a fixed embedding is polynomially equivalent to deciding whether a flat c-graph $C_f(G_f,T_f)$ admits a c-planar drawing where $G_f$ has a fixed embedding.

By the above observations some results on flat clustered graphs can be immediately exported to general c-graphs. Consider for example the following.

\begin{theorem}(\hspace{-0.1em}\cite[Theorem~1]{cdfk-atcefcg-14}).\label{th:brillo-original}
There exists an $O(n^3)$-time algorithm to test the c-planarity of an n-vertex embedded flat c-graph $C$ with at most two vertices per cluster on each face.
\end{theorem}

We generalize Theorem~\ref{th:brillo-original} to non-flat c-graphs.

\begin{theorem}\label{th:brillo}
Let $C(G,T)$ be an n-vertex c-graph where $G$ has a fixed embedding. There exists an $O(n^3 \cdot h(T)^3)$-time algorithm to test the c-planarity of $C$ if each lower cluster has at most two vertices on the same face of $G$ and each higher cluster has at most two inter-cluster edges on the same face of $G$.
\end{theorem}
\begin{sketch}
The proof is based on showing that, starting from a c-graph $C(G,T)$ that satisfies the hypotheses of the statement, the equivalent flat c-graph $C_f(G_f,$ $T_f)$ built as described in the proof of Theorem~\ref{th:main} satisfies the hypotheses of Theorem~\ref{th:brillo-original}. 
Hence, we first transform $C(G,T)$ into $C_f(G_f,T_f)$ in $O(n^2)$ time and then apply Theorem~\ref{th:brillo-original} to $C_f(G_f,T_f)$, which gives an answer to the c-planarity test in $O(n_f^3)$ time, which is, by Property~\ref{pr:vertices} of Lemma~\ref{le:properties}, $O(n^3 \cdot h(T)^3)$ time.    
\end{sketch}

In~\cite{degg-sfefcg-conf-17} it has been proven that \textsc{Flat Clustered Planarity} admits a subexponential-time algorithm when the underlying graph has a fixed embedding and its maximum face size~$\ell$ belongs to $o(n)$. 

\begin{theorem} (\hspace{-0.1em}\cite[Theorem 3]{degg-sfefcg-conf-17}).\label{th:giordano1}
\textsc{Flat Clustered Planarity} can be solved in $2^{O(\sqrt{\ell n}\cdot \log n)}$ time for n-vertex embedded flat c-graphs with maximum face size~$\ell$.
\end{theorem}



The authors of~\cite{degg-sfefcg-conf-17} ask whether their results can be generalized to non-flat c-graphs. We give an affirmative answer with the following theorem.

\begin{theorem}\label{th:giordano1-extension}
\textsc{Clustered Planarity} can be solved in $2^{O(h(T) \cdot \sqrt{\ell n}\cdot \log(n \cdot h(T))}$ time for n-vertex embedded c-graphs with maximum face size $\ell$ and height~$h(T)$ of the inclusion tree.
\end{theorem}

\begin{sketch}
The proof is based on applying Theorem~\ref{th:giordano1} to the equivalent flat c-graph $C_f(G_f,T_f)$ built as described in the proof of Theorem~\ref{th:main}. 
\end{sketch}

Observe that Theorem~\ref{th:giordano1-extension} gives a subexponential-time upper bound for \textsc{Clustered Planarity} whenever $\ell \cdot h(T)^2 \in o(n)$.
Also observe that Theorems~\ref{th:brillo} and~\ref{th:giordano1-extension} are actual generalizations of the corresponding Theorems~\ref{th:brillo-original} and~\ref{th:giordano1}, respectively, as they yield the same bounds when applied to flat clustered graphs.





\section{Proof of Theorem~\ref{th:main2}}\label{se:reduction2}

In this section we reduce \textsc{Flat Clustered Planarity} to \textsc{Independent Flat Clustered Planarity} by applying a transformation very similar to the one described in Section~\ref{se:reduction} to each non-independent cluster.  

Let $C(G,T)$ be a flat c-graph. Let $k$ be the number of lower clusters of $C$ that are not independent. The reduction consists of a sequence of transformations of $C=C_0$ into $C_1, C_2, \dots, C_{k}$ where each $C_i$, $i=0,\dots,k$, is a flat c-graph with $k-i$ non-independent lower clusters.   

\begin{figure}[tb]
\centering
\subfigure[]{\includegraphics[page=1,height=0.28\textwidth]{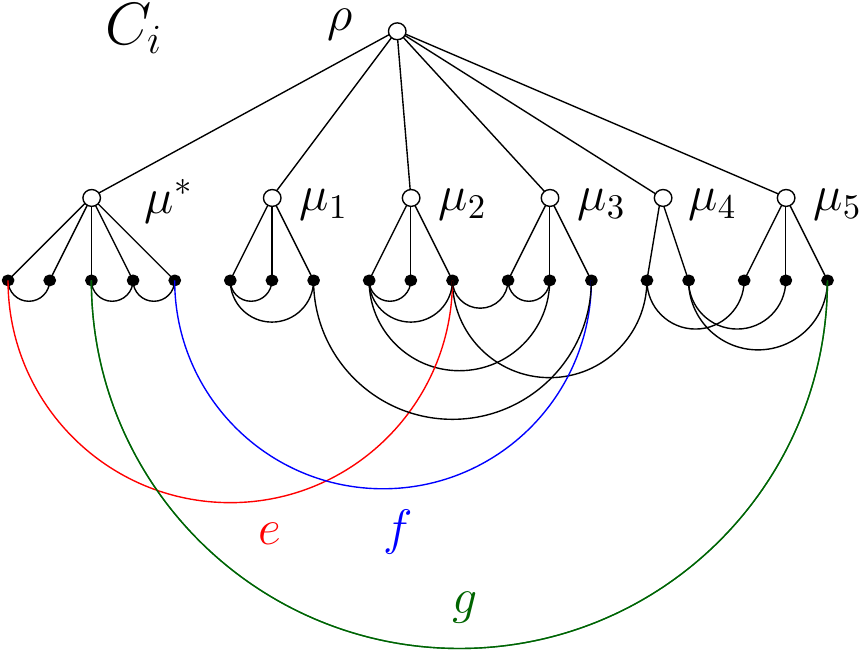}\label{fi:construction2-a}}
\hfil
\subfigure[]{\includegraphics[page=2,height=0.28\textwidth]{figures/construction2.pdf}\label{fi:construction2-b}}
\ifARXIV
\vspace{-0.02\textwidth}
\fi
\caption{\subref{fi:construction2-a} A flat c-graph $C_i$ with a non-independent cluster $\mu^*$. \subref{fi:construction2-b} The construction of $C_{i+1}$ where $\mu^*$ is replaced by independent clusters $\nu_1$, \dots, $\nu_5$, $\chi$, and $\varphi$.}\label{fi:construction2}
\end{figure}

Consider a flat c-graph $C_i(G_i,T_i)$, with $i=0, \dots, k-1$, such that $C_i$ has $k-i$ non-independent clusters and let $\mu^*$ be a non-independent cluster of $C$. We show how to construct an flat c-graph $C_{i+1}(G_{i+1},T_{i+1})$ equivalent to $C_i$ and such that $C_{i+1}$ has $k-i-1$ non-independent clusters (refer to Fig.~\ref{fi:construction2}).
Denote by $\mu_j$, with $j=1, 2, \dots, l$, those children of $r(T_i)$ such that $\mu_j \neq \mu^*$. Suppose that $\mu^*$ has children $v_1, v_2, \dots, v_h$, which are vertices of~$G_i$.

The underlying graph $G_{i+1}$ of $C_{i+1}$ is obtained from $G_i$ by introducing, for each inter-cluster edge $e=(u,v)$ of $\mu^*$, two new vertices $e_\chi$ and $e_\varphi$ and replacing $e$ with a path $(u,e_\chi)(e_\chi,e_\varphi)(e_\varphi,v)$. The inclusion tree $T_{i+1}$ of $C_{i+1}$ is obtained from $T_i$ by removing cluster $\mu^*$ and introducing, for each $j=1, 2, \dots, h$, a lower cluster $\nu_j$ child of $r(T_{i+1})$ containing only $v_j$. We also introduce two lower clusters $\chi$ and $\varphi$ as children of $r(T_{i+1})$ that contain all the vertices $e_\chi$ and $e_\varphi$, respectively, introduced when replacing each inter-cluster edge $e$ of $\mu^*$ with a path. 
It is easy to see that $C_{i+1}$ is a flat clustered graph and that it has one non-independent cluster less than $C_i$. 

We prove the following lemma assuming that Hypothesis \Hconnected holds. 
\ifARXIV
The complete proof is in~\ref{ap:reduction2}.
\else
The complete proof is in~\cite{cp-cpfcp-tr-18}.
\fi

\begin{figure}[tb]
\centering
\subfigure[]{\includegraphics[page=1,height=0.35\textwidth]{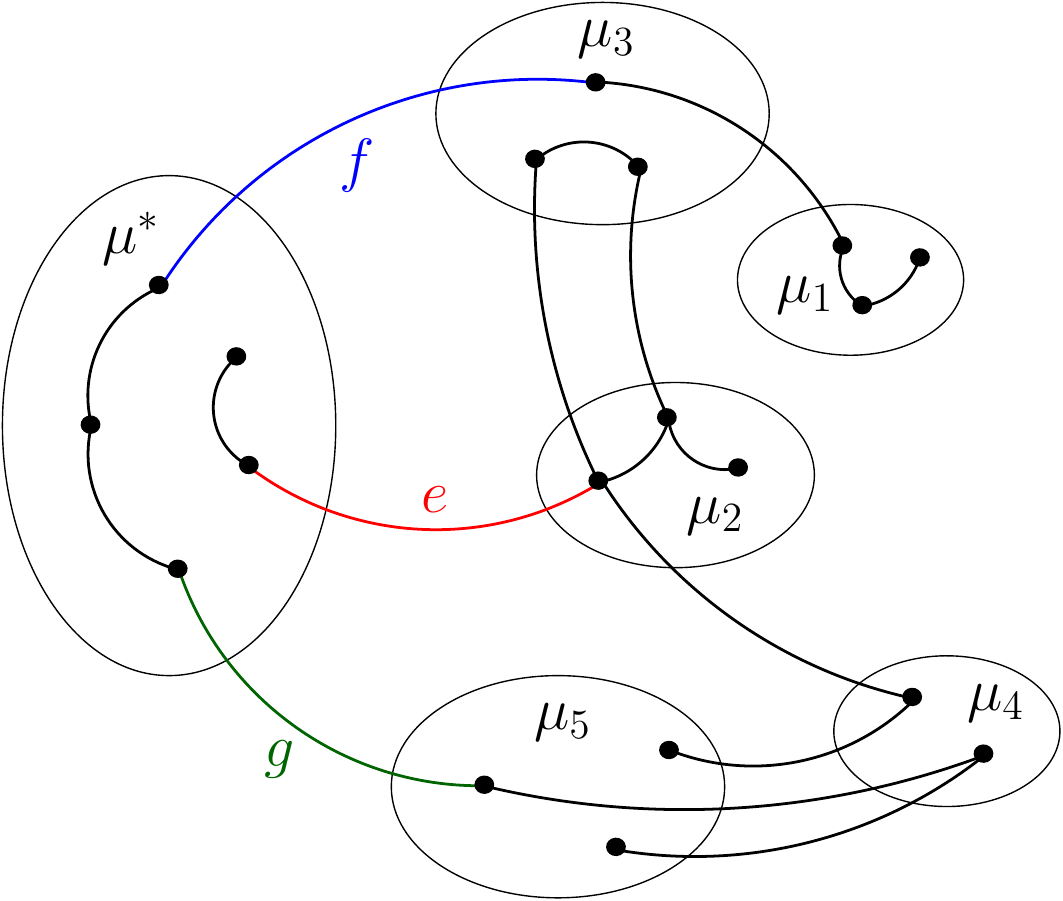}\label{fi:fc2ifc-a}}
\hfil
\subfigure[]{\includegraphics[page=2,height=0.35\textwidth]{figures/fc2ifc.pdf}\label{fi:fc2ifc-b}}
\caption{\subref{fi:fc2ifc-a} A c-planar drawing of the flat c-graph of Fig.~\ref{fi:construction2-a}. \subref{fi:fc2ifc-b} The corresponding c-planar drawing the flat c-graph of Fig.~\ref{fi:construction2-b} where the non-independent cluster $\mu^*$ is replaced by independent clusters $\nu_1$, \dots, $\nu_5$, $\chi$, and $\varphi$.}\label{fi:fc2ifc}
\end{figure}

\begin{lemma}\label{le:main2}
$C_i(G_i,T_i)$ is c-planar if and only if $C_{i+1}(G_{i+1},T_{i+1})$ is c-planar.
\end{lemma}

\begin{sketch}
The proof is very similar to the proof of Lemma~\ref{le:main}. First, we show that, given a c-planar drawing $\Gamma(C_i)$ of the flat c-graph $C_i$ it is easy to construct a c-planar drawing $\Gamma(C_{i+1})$ of $C_{i+1}$ (see, as an example, Fig.~\ref{fi:fc2ifc}). 
Second we show that, given a c-planar drawing $\Gamma(C_{i+1})$ of the flat c-graph $C_{i+1}$ it is possible to construct a c-planar drawing $\Gamma(C_{i})$ of $C_{i}$. 
This second part of the proof is complicated by the fact that, since in this case Hypothesis \Hnotroot does not apply, we may have that in $\Gamma(C_{i+1})$ the region $R(\varphi)$ is embraced by inter-cluster edges and region boundaries of $R(\nu_1)$, $R(\nu_2)$, \dots $R(\nu_l)$, and $R(\chi)$%
\ifARXIV
(see Fig.~\ref{fi:ifc2fc-a} in~\ref{ap:reduction2}). 
\else
. 
\fi
Hence, before identifying the region $R(\mu^*)$ the drawing $\Gamma(C_{i+1})$ needs to be modified so that the external face touches $R(\varphi)$. This can be easily done by 
\ifARXIV
rerouting edges (see the example in Fig.~\ref{fi:ifc2fc-b}). 
\else
rerouting edges. 
\fi
\end{sketch}

The proof of Theorem~\ref{th:main2} is concluded by showing that each $G_{i+1}$ can be obtained from $G_i$ in time proportional to the number of vertices and inter-cluster edges of~$\mu^*$, which gives an overall $O(n)$ time for the reduction. 

\section{Remarks about Theorem~\ref{th:main2}}\label{se:discussion2}

Starting from a flat c-graph, the reduction described in Section~\ref{se:reduction2} allows us to find an equivalent 
independent flat c-graph with the properties stated in the following lemma 
\ifARXIV
(the proof is in~\ref{ap:discussion2}). 
\else
(the proof can be found in~\cite{cp-cpfcp-tr-18}). 
\fi
\begin{lemma}\label{le:properties2}
Let $C_f(G_f,T_f)$ be an $n_f$-vertex flat clustered graph with $c_f$ clusters. The independent flat clustered graph $C_{\textrm{if}}(G_{\textrm{if}},T_{\textrm{if}})$ equivalent to $C_f$ built as described in the proof of Theorem~\ref{th:main2} has the following properties:
\begin{enumerate}
\item\label{pr:subdivision2} Graph $G_{\textrm{if}}$ is a subdivision of $G_f$
\item\label{pr:edges2} Each inter-cluster edge of $G_f$ is replaced by a path of length at most $4$.
\item\label{pr:vertices2} The number of vertices of $G_{\textrm{if}}$ is $O(n_f)$ 
\item\label{pr:clusters2} The number of clusters of $C_{\textrm{if}}$ (including the root) is $c_{\textrm{if}} \leq 2c_f + n_f -1$
\end{enumerate}
\end{lemma}

Also, a further property can be pursued.

\begin{observation}\label{ob:achieved}\hspace{-0.6em}\textbf{.}~
At the same asymptotic cost of the reduction described in the proof of Theorem~\ref{th:main2} it can be achieved that non-root clusters are of two types: (\textsc{Type~1}) clusters containing a single vertex of arbitrary degree or (\textsc{Type~2}) clusters containing multiple vertices of degree two. 
\end{observation}

All observations of Section~\ref{se:discussion} regarding the consequences of Property~\ref{pr:subdivision} of Lemma~\ref{le:properties} apply here to of Property~\ref{pr:subdivision2} of Lemma~\ref{le:properties2}. Further, the two reductions can be concatenated yielding the following.

\begin{lemma}\label{le:properties3}
Let $C(G,T)$ be an $n$-vertex clustered graph with $c$ clusters. The independent flat clustered graph $C_{\textrm{if}}(G_{\textrm{if}},T_{\textrm{if}})$ equivalent to $C$ built by concatenating the reduction of Theorem~\ref{th:main} and the reduction of Theorem~\ref{th:main2}, as modified by Observation~\ref{ob:achieved}, has the following properties:
\begin{enumerate}
\item\label{pr:subdivision3} Graph $G_{\textrm{if}}$ is a subdivision of $G$
\item\label{pr:edges3} Each inter-cluster edge of $G_f$ is replaced by a path of length at most $4h(T)-4$
\item\label{pr:vertices3} The number of vertices of $G_{\textrm{if}}$ is $O(n^2)$ 
\item\label{pr:clusters3} The number of clusters of $C_{\textrm{if}}$ is $O(n \cdot h(T))$
\item\label{pr:cluster-types3} Non-root clusters are of two types: (\textsc{Type~1}) clusters containing a single vertex of arbitrary degree or (\textsc{Type~2}) clusters containing multiple vertices of degree two 
\end{enumerate}
\end{lemma}

Lemma~\ref{le:properties3} describes the most constrained version of \textsc{Clustered Planarity} that is known to be polynomially equivalent to the general problem.
Observe that if all non-root clusters of a c-graph $C(G,T)$ are of \textsc{Type~1} 
then \textsc{Independent Flat Clustered Planarity} is linear, since $C$ is c-planar if and only if $G$ is planar. Conversely, if all clusters are of \textsc{Type~2} then the underlying graph is 
a collection of cycles, and the problem has unknown complexity~\cite{cdpp-cccc-05,cdpp-ecpg-09}.

\section{Conclusions and Open Problems}\label{se:conclusions}

We showed that \textsc{Clustered Planarity} can be reduced to \textsc{Flat Clustered Planarity} and that this problem, in turn, can be reduced to \textsc{Independent Flat Clustered Planarity}. The consequences of these results are twofold: on one side the investigations about the complexity of \textsc{Clustered Planarity} could legitimately be restricted to (independent) flat clustered graphs, neglecting more complex hierarchies of the inclusion tree; on the other side some polynomial-time results on flat clustered graphs could be easily exported to general c-graphs (we gave some examples in Section~\ref{se:discussion}). 

We remark that while Theorems~\ref{th:main} and~\ref{th:main2} are formulated in terms of decision problems, their proofs offer a solution of the corresponding search problems, meaning that they actually describe a polynomial-time algorithm to compute a c-planar drawing of a c-graph, provided to have a c-planar drawing of the corresponding flat c-graph or a c-planar drawing of the corresponding independent flat c-graph.  
 
Several interesting questions are left open:

\begin{itemize}
\item Can the reduction presented in this paper be used to generalize some other polynomial-time testing algorithm for \textsc{Flat Clustered Planarity} to plain \textsc{Clustered Planarity}?
\item What is the complexity of \textsc{Independent Flat Clustered Planarity} when the underlying graph is a cycle? We know that this problem is polynomial only for constrained drawings of the inter-cluster edges~\cite{cdpp-cccc-05,cdpp-ecpg-09}.
\item What is the complexity of \textsc{Independent Flat Clustered Planarity} when the number of \textsc{Type 2} clusters is bounded? 
\end{itemize}

\clearpage
\bibliographystyle{splncs03}
\bibliography{abbrv,bibliography}

\ifARXIV  

\clearpage
\appendix

\newgeometry{left=2cm,right=2cm,top=2cm,bottom=2cm}

\renewcommand{\thesection}{Appendix~\Alph{section}}

\section{-- Proof of Lemmas~\ref{le:homogeneous-subtree} and~\ref{le:preconditions} of Section~\ref{se:preliminaries}.}\label{ap:lemmata}

\noindent
\textbf{Lemma~\ref{le:homogeneous-subtree}.} \emph{
A homogeneous tree $T$ of height $h(T) \geq 2$ and size $\mathcal{S}(T)>0$ contains at least one node $\mu^* \neq r(T)$ such that $T[\mu^*]$ is flat.
}
\begin{proof}
Let $v$ be a leaf of $T$ whose depth is $h(T)$. Since $\mathcal{S}(T)>0$, $T$ is not flat and $h(T) \geq 3$. Consider the parent $\mu^*$ of the parent $\mu$ of $v$. The subtree $T[\mu^*]$ of $T$ has $h(T[\mu^*])=2$. Since $T$ is homogeneous and $\mu^*$ has a non-leaf child $\mu$, all the children of $\mu^*$ are internal nodes. Hence, $T[\mu^*]$ is flat.  
\end{proof}

\noindent
\textbf{Lemma~\ref{le:preconditions}.} \emph{
An instance $C(G,T)$ of \textsc{Clustered Planarity} with $n$ vertices and $c$ clusters can be reduced in time $O(n+c)$ to an equivalent instance such that:
\begin{enumerate}[(1)]
\item $T$ is homogeneous,
\item $r(T)$ has at least two children, and 
\item $h(T) \leq n-1$.
\end{enumerate}
}

\begin{proof}
First we prove Property (\ref{le:preconditions1}). Suppose that the inclusion tree $T$ of a c-graph $C(G,T)$ is not homogeneous. We transform $C$ into an equivalent c-graph $C_h(G,T_h)$ such that $T_h$ is homogeneous. Consider a node $\mu^*$ that has both internal node children and leaf children $v_1$, $v_2$, \dots, $v_k$. For each such $\mu^*$ and for each child $v_i$ of $\mu^*$, we insert between $\mu^*$ and $v_i$ a lower node $\mu_i$ that is child of $\mu^*$ and parent of $v_i$. The obtained c-graph $C_h(G,T_h)$ is homogeneous and may be constructed in time $O(n+c)$.
Also, given a c-planar drawing $\Gamma_h(C_h)$ of $C_h$ one can immediately obtain a c-planar drawing $\Gamma(C)$ of $C$ simply by ignoring the boundaries of the regions $R(\mu_i)$, where $\mu_i$ is a cluster introduced by the above described transformation. Conversely, given a c-planar drawing $\Gamma(C)$ of $C$ one can obtain a c-planar drawing $\Gamma_h(C_h)$ of $C_h$ by inserting a small boundary around the vertices $v_i$ that changed their parent in the above transformation. Hence, $C(G,T)$ is c-planar if and only if $C_h(G,T_h)$ is c-planar.  

Finally, we prove Properties~(\ref{le:preconditions2}) and~(\ref{le:preconditions3}). Suppose to have an instance $C(G,T)$ with $h(T) > n-1$ or such that $r(T)$ has a single child. 
We traverse $T$ and recursively replace each cluster that has a single child with its child. The obtained instance $C'(G,T')$ is equivalent to original one since from a c-planar drawing of $C$ one can obtain a c-planar drawing of $C'$ simply by ignoring the boundaries of the removed clusters and from a c-planar drawing of $C'$ one can obtain a c-planar drawing of $C$ by suitably adding the boundary of each removed parent cluster around the boundary of its child cluster (or around the child vertex leaf). Since all clusters of $C'$ have at least two children Property~(\ref{le:preconditions2}) is satisfied. We claim that $h(T') \leq n$. In fact, since all the $n_i$ internal nodes of $T'$ have degree at least two, we have that the number $n$ of leaves of $T'$ is at least $n \geq n_i + 1$. Hence, $h(T') \leq n_i \leq n-1$. This proves Property~(\ref{le:preconditions3}).
\end{proof}


\section{-- Proofs of Lemmas~\ref{le:homogeneous-preserved}--\ref{le:main} of Section~\ref{se:reduction}.}\label{ap:lemma}

Let $C_i(G_i,T_i)$ be a flat c-planar c-graph and let $\mu^* \neq r(T_i)$ be a node of $T_i$ such that $T_i[\mu^*]$ is flat. Denote by $\nu_1, \nu_2, \dots, \nu_h$ the children of $\mu^*$ and by $\mu_1$, $\mu_2$, \dots, $\mu_k$ the siblings of $\mu^*$ in $T_i$. Let $C_{i+1}(G_{i+1},T_{i+1})$ be the flat c-graph constructed as described in Section~\ref{se:reduction}. We prove the following lemmas. 

\bigskip
\noindent
\textbf{Lemma~\ref{le:homogeneous-preserved}.} \emph{If $T_i$ is homogeneous then $T_{i+1}$ is homogeneous.}
\begin{proof}
Suppose $T_i$ is homogeneous. The only part of $T_i$ that is changed in~$T_{i+1}$ is the subtree $T_i[\nu]$. In particular, $\nu$ has all cluster children, while the newly introduced clusters $\chi$ and $\varphi$ have all vertex children. Hence $T_{i+1}$ is homogeneous. 
\end{proof}

\medskip\noindent
\textbf{Lemma~\ref{le:size-decreases}.} \emph{We have that $\mathcal{S}(T_{i+1}) = \mathcal{S}(T_i) - 1$.}
\begin{proof}
Consider a node $\mu \neq r(T_{i+1})$ of $T_{i+1}$ for which $h(T_{i+1}[\mu]) > 1$. Since the transformation of $C_i$ into $C_{i+1}$ only reduces the height of some subtree of $T_i$, such a node was also the root of a subtree of height greater than $1$ in~$T_i$. Conversely, consider a node $\mu \neq r(T_{i})$ of $T_{i}$ for which $h(T_{i+1}[\mu]) > 1$. If $\mu = \mu^*$ then $\mu$ is not present in~$T_{i+1}$, otherwise it is still the root of a subtree of height greater than one in~$T_{i+1}$. Hence, the number of nodes that are root of subtrees of height greater than one of $T_{i+1}$ is reduced by one with respect to the same number in~$T_i$.
\end{proof}

\bigskip
\noindent
\textbf{Lemma~\ref{le:final-flat}.} \emph{The c-graph $C_f=C_{\mathcal{S}(T)}$ is flat.}

\begin{proof}
By Property~\ref{le:preconditions1} of Lemma~\ref{le:preconditions} we can assume that $T$ is homogeneous. Lemma~\ref{le:homogeneous-preserved} ensures that all $C_i$, with $i=1, \dots, {\mathcal{S}(T)}$ are also homogeneous. By Lemma~\ref{le:size-decreases} we have that the sizes of the trees $T_i$ are decreasing and, in particular, that the size of $T_{\mathcal{S}(T)}$ is $\mathcal{S}(T) - \mathcal{S}(T) = 0$. Therefore, $T_{\mathcal{S}(T)}$ is a homogeneous tree that has size $0$ and, hence, is a flat tree. 
\end{proof}

Now, we provide the proof of Lemma~\ref{le:main} in the general case, i.e., without leveraging on Hypotheses \Hconnected and \Hnotroot. 

\bigskip
\noindent
\textbf{Lemma~\ref{le:main}.} \emph{$C_i(G_i,T_i)$ is c-planar if and only if $C_{i+1}(G_{i+1},T_{i+1})$ is c-planar.}

\begin{proof}
The first direction of the proof is straightforward. Let $\Gamma(C_i)$ be a c-planar drawing of $C_i$. We show how to construct a c-planar drawing $\Gamma(C_{i+1})$ of~$C_{i+1}$. Consider the region $R(\mu^*)$ that contains $R(\nu_i)$, with $i=1,\dots,h$ (refer to Fig.~\ref{fi:c2fc1}). The boundary of $R(\mu^*)$ is crossed exactly once by each inter-cluster edge of $\mu^*$. Identify outside the boundary of $R(\mu^*)$ two arbitrarily thin regions $R(\chi)$ and $R(\varphi)$ that follow the boundary of $R(\mu^*)$ and that intersect all and only the inter-cluster edges of $\mu^*$ exactly once (see Fig.~\ref{fi:c2fc2}). Insert into each inter-cluster edge $e$ of $\mu^*$ two vertices $e_\chi$ and $e_\varphi$, placing $e_\chi$ inside $R(\chi)$ and $e_\varphi$ inside $R(\varphi)$. By ignoring $R(\mu^*)$ you have a c-planar drawing 
$\Gamma(C_{i+1})$ of $C_{i+1}$.

\begin{figure}[htb]
\centering
\subfigure[]{\includegraphics[page=1,width=0.48\textwidth]{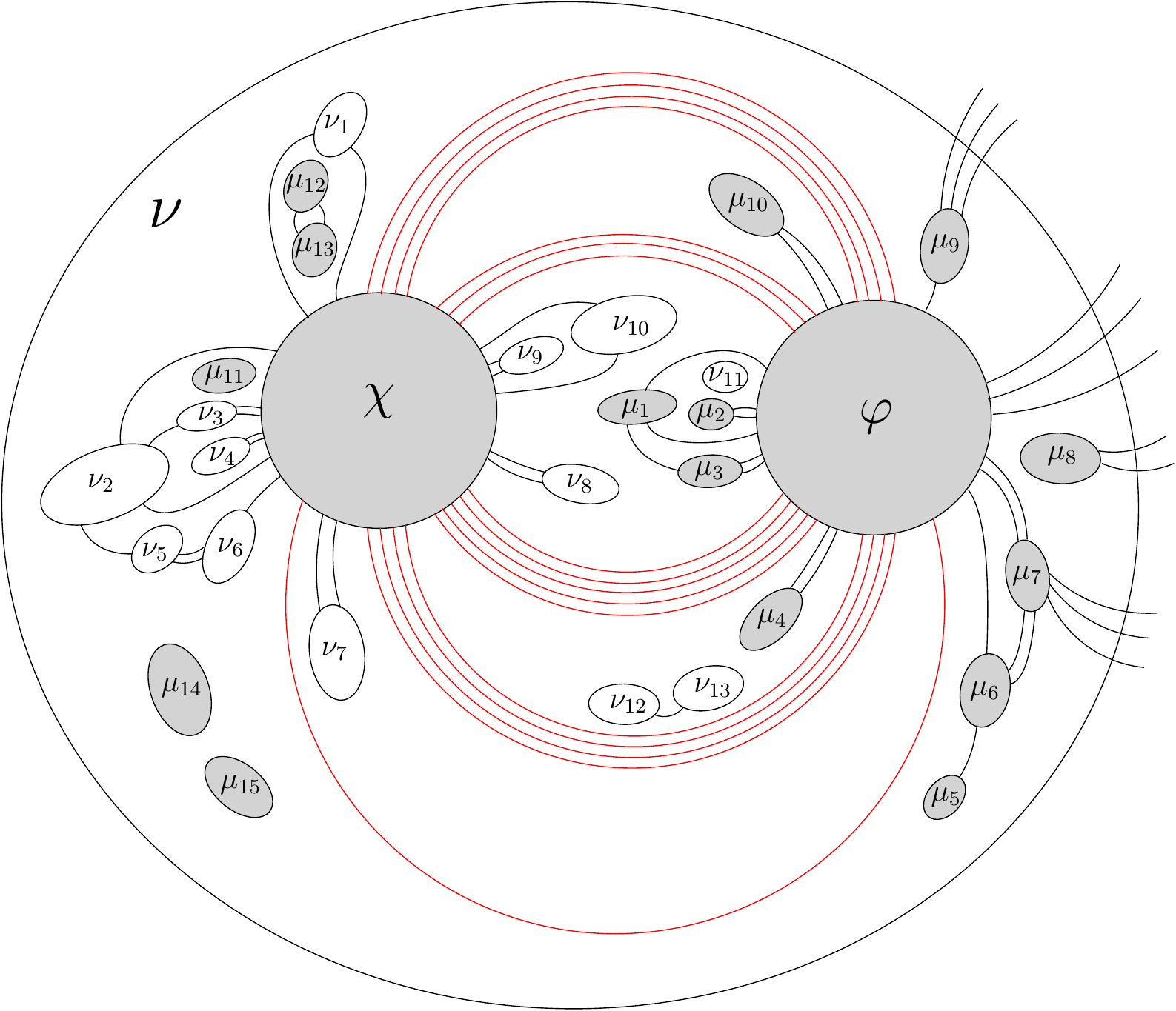}\label{fi:not-connected-a}}
\hfil
\subfigure[]{\includegraphics[page=3,width=0.48\textwidth]{figures/not-connected.pdf}\label{fi:not-connected-b}}
\caption{\subref{fi:not-connected-a} A possible drawing of cluster $\nu$ in $\Gamma(C_{i+1})$ in the case of non-connected $G_{i+1}$. The inter-cluster edges between $\chi$ and $\varphi$ are drawn red. \subref{fi:not-connected-b} The same drawing after the removal of the floating regions.}\label{fi:not-connected}
\end{figure}

Conversely, suppose to have a c-planar drawing $\Gamma(C_{i+1})$ of $C_{i+1}$. We show how to construct a c-planar drawing $\Gamma(C_i)$ of~$C_i$. Consider the regions $R(\chi)$ and $R(\varphi)$ inside $R(\nu)$ (refer to Fig.~\ref{fi:not-connected-a}). Regions $R(\chi)$ and $R(\varphi)$ are joined by the $p$ inter-cluster edges (drawn red in Fig.~\ref{fi:not-connected-a}) introduced when replacing each inter-cluster edge $e_i$ of $\mu^*$, where $i=1,\dots,p$, with a path. Such inter-cluster edges of $\chi$ and $\varphi$ partition $R(\nu)$ into $p$ regions that have to host the remaining children of $\nu$ and the inter-cluster edges among them. In particular, $p-1$ of these regions are simple and bounded by two inter-cluster edges and two portions of the boundaries of $R(\chi)$ and $R(\varphi)$. One of such regions, instead, is also externally bounded by the boundary of $R(\nu)$.

Now, consider the regions corresponding to the children $\nu_i$ of $\nu$, with $i=1,\dots,h$, that were originally children of $\mu^*$. These regions (filled white in Fig.~\ref{fi:not-connected-a}) may be without any inter-cluster edge (as, for example, $R(\nu_{11})$ in Fig.~\ref{fi:not-connected-a}); may have inter-cluster edges among themselves (as, for example, $R(\nu_{2})$, $R(\nu_{3})$, $R(\nu_{4})$, $R(\nu_{5})$, $R(\nu_{6})$, $R(\nu_{12})$, and $R(\nu_{13})$ in Fig.~\ref{fi:not-connected-a}); and may be connected to $R(\chi)$ (as, for example, $R(\nu_{1})$, $R(\nu_{2})$, $R(\nu_{3})$, $R(\nu_{4})$, $R(\nu_{6})$, $R(\nu_{7})$, $R(\nu_{8})$, $R(\nu_{9})$, and $R(\nu_{10})$ in Fig.~\ref{fi:not-connected-a}). However, by construction these regions cannot have inter-cluster edges connecting them to $R(\varphi)$, or connecting them to the regions of the original children $\mu_i$ of $\nu$, or exiting the border of $R(\nu)$. Hence, the regions corresponding to $\nu_1$, \dots, $\nu_h$ can be classified into two sets, denoted $\mathcal{A}_\chi$ and $\mathcal{F}_\chi$, of `anchored regions' and `floating regions' of $\chi$, respectively, where an \emph{anchored region of $\chi$} is a region $R(\nu_a)$ whose cluster $\nu_a$ contains at least one vertex of $G_{i+1}$ that is connected (via a path) to a vertex in $\chi$ and a \emph{floating region of $\chi$} is a region $R(\nu_f)$ whose cluster $\nu_f$ contains all vertices not connected to vertices in $\chi$. For example, in Fig.~\ref{fi:not-connected-a}, $\mathcal{F}_\chi$ contains $R(\nu_{11})$, $R(\nu_{12})$, and $R(\nu_{13})$, while $\mathcal{A}_\chi$ contains all the other white-filled regions. 

Analogously, consider the regions $R(\mu_j)$, with $j=1, \dots, k$, corresponding to the original children $\mu_j \neq \mu^*$ of $\nu$ (filled gray in Fig.~\ref{fi:not-connected-a}). These regions may be without any inter-cluster edge (as, for example, $R(\mu_{11})$, $R(\mu_{14})$, and $R(\mu_{15})$ in Fig.~\ref{fi:not-connected-a}); may have inter-cluster edges among themselves (as, for example, $R(\mu_{1})$, $R(\mu_{3})$, $R(\mu_{5})$, $R(\mu_{6})$, $R(\mu_{7})$, $R(\mu_{12})$, and $R(\nu_{13})$ in Fig.~\ref{fi:not-connected-a}); may have inter-cluster edges connecting them to $R(\varphi)$ (as, for example, $R(\mu_{2})$, $R(\mu_{3})$, $R(\mu_{4})$, $R(\mu_{6})$, $R(\mu_{7})$, $R(\mu_{9})$, and $R(\mu_{10})$ in Fig.~\ref{fi:not-connected-a}); or may have inter-cluster edges connecting them the rest of the graph outside $R(\nu)$ (as, for example, $R(\mu_{8})$ in Fig.~\ref{fi:not-connected-a}). However, by construction these regions cannot have inter-cluster edges connecting them to $R(\chi)$, or connecting them to the the regions in $\mathcal{F}_\chi$ or $\mathcal{A}_\chi$. Hence, we can classify the regions corresponding to $\mu_1$, \dots, $\mu_k$ into two sets, denoted $\mathcal{A}_\varphi$ and $\mathcal{F}_\varphi$, of `anchored regions' and `floating regions' of $\varphi$, where an \emph{anchored region of $\varphi$} is a region $R(\mu_a)$ whose cluster $\mu_a$ contains at least one vertex of $G_{i+1}$ that is connected to a vertex in $\varphi$ or to a vertex outside $\nu$ and a \emph{floating region of $\varphi$} is a region $R(\mu_f)$ whose cluster $\mu_f$ contains all vertices not connected to vertices in $\varphi$ nor outside $\nu$. For example, in Fig.~\ref{fi:not-connected-a}, set $\mathcal{F}_\varphi$ contains $R(\mu_{11})$, $R(\mu_{12})$, $R(\mu_{13})$, $R(\mu_{14})$, and $R(\mu_{15})$, while $\mathcal{A}_\varphi$ contains all the other gray-filled regions.

Our strategy will be that of removing altogether from $\Gamma(C_{i+1})$ the drawings of the floating regions (and all their content), possibly modifying the drawing of the remaining graph, and then suitably reinserting the drawing of the floating regions. 

Suppose now to have temporarily removed from $\Gamma(C_{i+1})$ the drawings of the floating regions in $\mathcal{F}_\chi$ and $\mathcal{F}_\varphi$ (see, for example, Fig.~\ref{fi:not-connected-b}). 

We define an auxiliary multigraph $H$ that has one vertex $v_\chi$ representing $\chi$ and one vertex $v_{\nu_i}$ for each child $\nu_i$ of $\mu^*$ such that $R(\nu_i) \in \mathcal{A}_\chi$. For each inter-cluster edge between two clusters $\lambda_1$ and $\lambda_2$ corresponding to the vertices $v_{\lambda_1}$ and $v_{\lambda_2}$ of $H_\chi$, respectively, we add an edge $(v_{\lambda_1},v_{\lambda_2})$ to $H$. Observe that $H$, by the definition of the anchored regions in $\mathcal{A}_\chi$, is connected.

Drawing $\Gamma(C_{i+1})$ induces a drawing $\Gamma(H)$ of the multigraph $H$, where each vertex $v_\lambda$ of $H$ is represented by the region $R(\lambda)$ of the cluster $\lambda$ corresponding to $v_\lambda$ and each edge $(v_{\lambda_1},v_{\lambda_2})$ of $H$ is represented as the corresponding inter-cluster edge of $\lambda_1$ and $\lambda_2$ restricted to the portion that is drawn outside the boundaries of $R(\lambda_1)$ and $R(\lambda_2)$. 

There are two cases: either $\Gamma(H)$ does not contain in one of its internal faces $R(\varphi)$ (\textsc{Case 1}, depicted in Fig.~\ref{fi:not-connected-b}) or it contains $R(\varphi)$ (\textsc{Case 2}, depicted in Fig.~\ref{fi:root-case-a}). 

In \textsc{Case 1} no change has to be done to $\Gamma(C_{i+1})$. In \textsc{Case 2} we modify $\Gamma(H)$ and, consequently, $\Gamma(C_{i+1})$ so to fall again into \textsc{Case 1}. Namely, we identify a minimal set $\{e_1, e_2, \dots, e_q\}$ of edges of $H$ that, if removed, would bring $R(\varphi)$ on the external face of $\Gamma(H)$ (for example in Fig~\ref{fi:root-case-a} this set contains only edge $e_1$). Starting from edge $e_1$, that is incident to the external face of $\Gamma(H)$, we redraw each $e_i$, with $i = 1, \dots, q$, as follows. Suppose that the curve for $e_i=(v_{\lambda_1},v_{\lambda_2})$ in $\Gamma(H)$ starts from a point $p_1$ on the boundary of $R(\lambda_1)$ and ends with a point $p_2$ on the boundary of $R(\lambda_2)$. We arbitrarily choose two distinct points $p_3$ and $p_4$, encountered in this order when traversing $e_i$ from $p_1$ to $p_2$. We remove the portion of $e_i$ between $p_3$ and $p_4$ and we redraw it by returning back from $p_3$ towards $p_1$ on the external face of $\Gamma(H)$ and then moving along the external face of $\Gamma(H)$ until we reach $p_4$ (see, for example, Fig.~\ref{fi:root-case-b}). Observe that this corresponds to moving the external face of $\Gamma(H)$ to a face that was previously an internal face of $\Gamma(H)$ enclosed by $e_i$. We carry on doing the same operation for each $e_i$, with $i=1, \dots, q$, until the external face of $\Gamma(H)$ is incident on the boundary of $R(\varphi)$. At this point we are in \textsc{Case 1}.

Observe that since $\Gamma(H)$ does not contain in one of its internal faces $R(\varphi)$, then it cannot contain any region in $\mathcal{A}_\varphi$ either, as, by definition, these regions are either connected to $R(\varphi)$ or to the boundary of $R(\nu)$. Hence, the internal faces of $\Gamma(H)$ only contain vertices and edges that in $C_i$ belong to $\mu^*$. 

Now we reinsert the drawings of the floating regions. We identify an arbitrarily small empty disk $F_\chi$ inside $R(\chi)$ and move inside $F_\chi$ the (suitably scaled down) drawings of the floating regions in $\mathcal{F}_\chi$. Analogously, we identify an arbitrarily small empty disk $F_\varphi$ inside $R(\varphi)$ and move inside $F_\varphi$ the (suitably scaled down) drawings of the floating regions in $\mathcal{F}_\varphi$. Consider the region $R(\mu^*)$ that is the region covered by $\Gamma(H)$. Such a region is connected, is simple, contains only vertices and nodes of $\mu^*$, and its boundary is a simple curve (see Fig.~\ref{fi:root-case-b}). Therefore, by neglecting the boundaries of $R(\chi)$ and $R(\varphi)$ and by removing their internal vertices and joining their incident edges we obtain a c-planar drawing $\Gamma(C_i)$ of~$C_i$.   
\end{proof}

\begin{figure}[htb]
\centering
\subfigure[]{\includegraphics[page=1,width=0.48\textwidth]{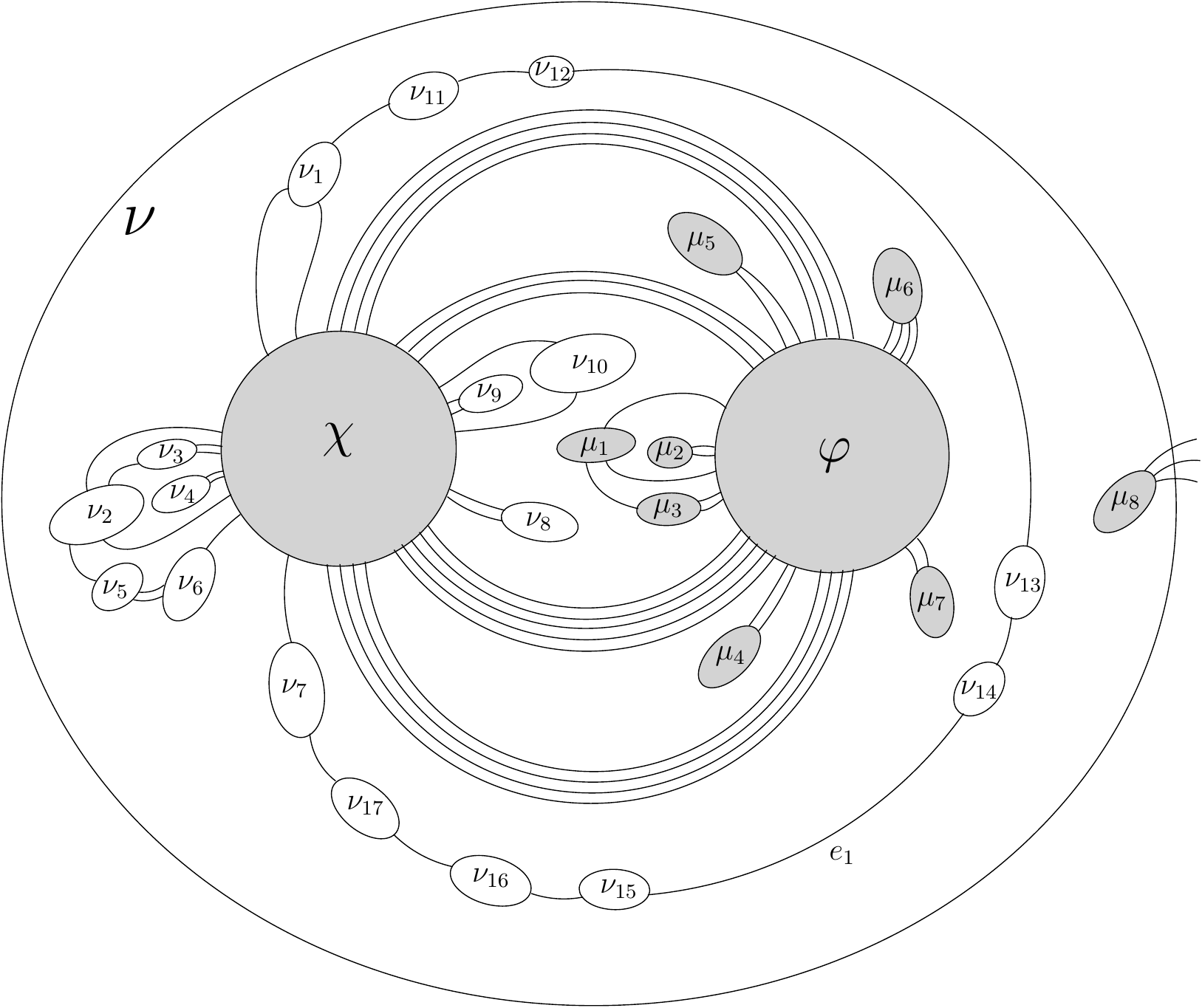}\label{fi:root-case-a}}
\hfil
\subfigure[]{\includegraphics[page=2,width=0.48\textwidth]{figures/root-case.pdf}\label{fi:root-case-b}}
\caption{\subref{fi:root-case-a} A possible drawing of cluster $\nu$ in $\Gamma(C_{i+1})$ in the case of non-connected $G_{i+1}$. \subref{fi:root-case-b} The corresponding drawing of $\nu$ in $\Gamma(C_i)$.}\label{fi:root-case}
\end{figure}

\section{-- Proof of Theorems~\ref{th:brillo} and~\ref{th:giordano1-extension} of Section~\ref{se:discussion}.}\label{ap:discussion}

\bigskip
\noindent
\textbf{Theorem~\ref{th:brillo}.} \emph{
Let $C(G,T)$ be an n-vertex c-graph where $G$ has a fixed embedding. There exists an $O(n^3 \cdot h(T)^3)$-time algorithm to test the c-planarity of $C$ if each lower cluster has at most two vertices on the same face of $G$ and each higher cluster has at most two inter-cluster edges on the same face of $G$.
}

\begin{proof}
The proof is based on showing that, starting from a c-graph $C(G,T)$ that satisfies the hypotheses of the statement, the equivalent flat c-graph $C_f(G_f,T_f)$ built as described in the proof of Theorem~\ref{th:main} satisfies the hypotheses of Theorem~\ref{th:brillo-original}. 
By Property~\ref{le:preconditions1} of Lemma~\ref{le:preconditions} we can assume that $T$ is homogeneous. Observe that the transformation of $T$ into an homogeneous tree described in the proof of Lemma~\ref{le:preconditions} only introduces lower clusters that contain a single vertex and, hence, preserves the property that each higher cluster has at most two inter-cluster edges incident to the same face.

The transformation of $G_i$ into $G_{i+1}$ described in the proof of Theorem~\ref{th:main} removes one higher cluster $\mu^*$ and introduces two lower clusters $\chi$ and $\varphi$. Each inter-cluster edge $e=(u,v)$ of $\mu^*$ is subdivided into three edges $(u,e_\chi)$, $(e_\chi,e_\varphi)$, and $(e_\varphi,v)$, where $e_\chi \in \chi$ and $e_\varphi \in \varphi$. Since at most two inter-cluster edges of $\mu^*$ belong to the same face of $G_i$ we have that the lower clusters $\chi$ and $\varphi$ have at most two vertices on the same face of $G_{i+1}$. Any other higher or lower cluster of $T_{i+1}$ is not modified by the transformation. It follows that $C_f(G_f,T_f)$, which, with the exception of the root cluster, has only lower clusters, satisfies the conditions of Theorem~\ref{th:brillo-original}. 

Hence, we first transform $C(G,T)$ into $C_f(G_f,T_f)$ in $O(n^2)$ time (Theorem~\ref{th:main}) and then apply Theorem~\ref{th:brillo-original} to $C_f(G_f,T_f)$, which gives an answer to the c-planarity test in $O(n_f^3)$ time, which is, by Property~\ref{pr:vertices} of Lemma~\ref{le:properties}, $O(n^3 \cdot h(T)^3)$ time.    
\end{proof}

\bigskip
\noindent
\textbf{Theorem~\ref{th:giordano1-extension}.} \emph{
\textsc{Clustered Planarity} can be solved in $2^{O(h(T) \cdot \sqrt{\ell n}\cdot \log(n \cdot h(T))}$ time for n-vertex embedded c-graphs with maximum face size $\ell$ and height~$h(T)$ of the inclusion tree.
}

\begin{proof}
The proof is based on applying Theorem~\ref{th:giordano1} to the flat c-graph $C_f(G_f,T_f)$ built as described in the proof of Theorem~\ref{th:main} and equivalent to $C(G,T)$. By Property~\ref{pr:edges} of Lemma~\ref{le:properties} each edge of $G$ is replaced by a path of length at most $2h(T)-2$. Hence, each face of $G_f$ has a maximum size $\ell_f = \ell \cdot O(h(T))$. Also, by Property~\ref{pr:vertices} of Lemma~\ref{le:properties} we have that the number of vertices of $G_f$ is $n_f \in O(n \cdot h(T))$. Theorem~\ref{th:giordano1} guarantees that we can test for c-planarity in $2^{O(\sqrt{\ell_f n_f}\cdot \log n_f)}$ time, which gives the statement.
\end{proof}

\section{-- Proof of Lemma~\ref{le:main2} of Section~\ref{se:reduction2}.}\label{ap:reduction2}

In this section, we provide the proof of Lemma~\ref{le:main2} in the general case, i.e., without leveraging on Hypothesis \mbox{\Hconnected}.
Let $C_i(G_i,T_i)$ be a flat c-planar c-graph and let $\mu^*$ be a non-independent cluster of $C_i$ containing vertices $v_1$, $v_2$, \dots, $v_h$ of $G_i$. Also, denote by $\nu_j$, with
$j = 1, 2, \dots, l$, those children of $r(T_i)$ such that $\nu_j \neq \mu^∗$. 
Let $C_{i+1}(G_{i+1},T_{i+1})$ be the flat c-graph constructed as described in Section~\ref{se:reduction2}. We have the following. 

\bigskip
\noindent
\textbf{Lemma~\ref{le:main2}.} \emph{$C_i(G_i,T_i)$ is c-planar if and only if $C_{i+1}(G_{i+1},T_{i+1})$ is c-planar.}

\begin{proof}
The proof is similar to the proof of Lemma~\ref{le:main}.
Given a c-planar drawing $\Gamma(C_i)$ of the flat c-graph $C_i$, we show how to construct a c-planar drawing $\Gamma(C_{i+1})$ of $C_{i+1}$ (refer to Fig.~\ref{fi:fc2ifc}). The construction is based on identifying two arbitrary thin regions $R(\chi)$ and $R(\varphi)$ outside the border of $R(\mu^*)$ such that $R(\chi)$ and $R(\varphi)$ intersect exactly once all and only the inter-cluster edges of $\mu^*$. By ignoring $R(\mu^*)$, by inserting for each inter-cluster edge $e$ of $\mu^*$ vertices $e_\chi$ and $e_\varphi$ in $R(\chi)$ and $R(\varphi)$, respectively, and by adding a boundary to each vertex $v_1$, \dots, $v_h$, we obtain $\Gamma(C_{i+1})$.

\begin{figure}[tb]
\centering
\subfigure[]{\includegraphics[page=1,height=0.35\textwidth]{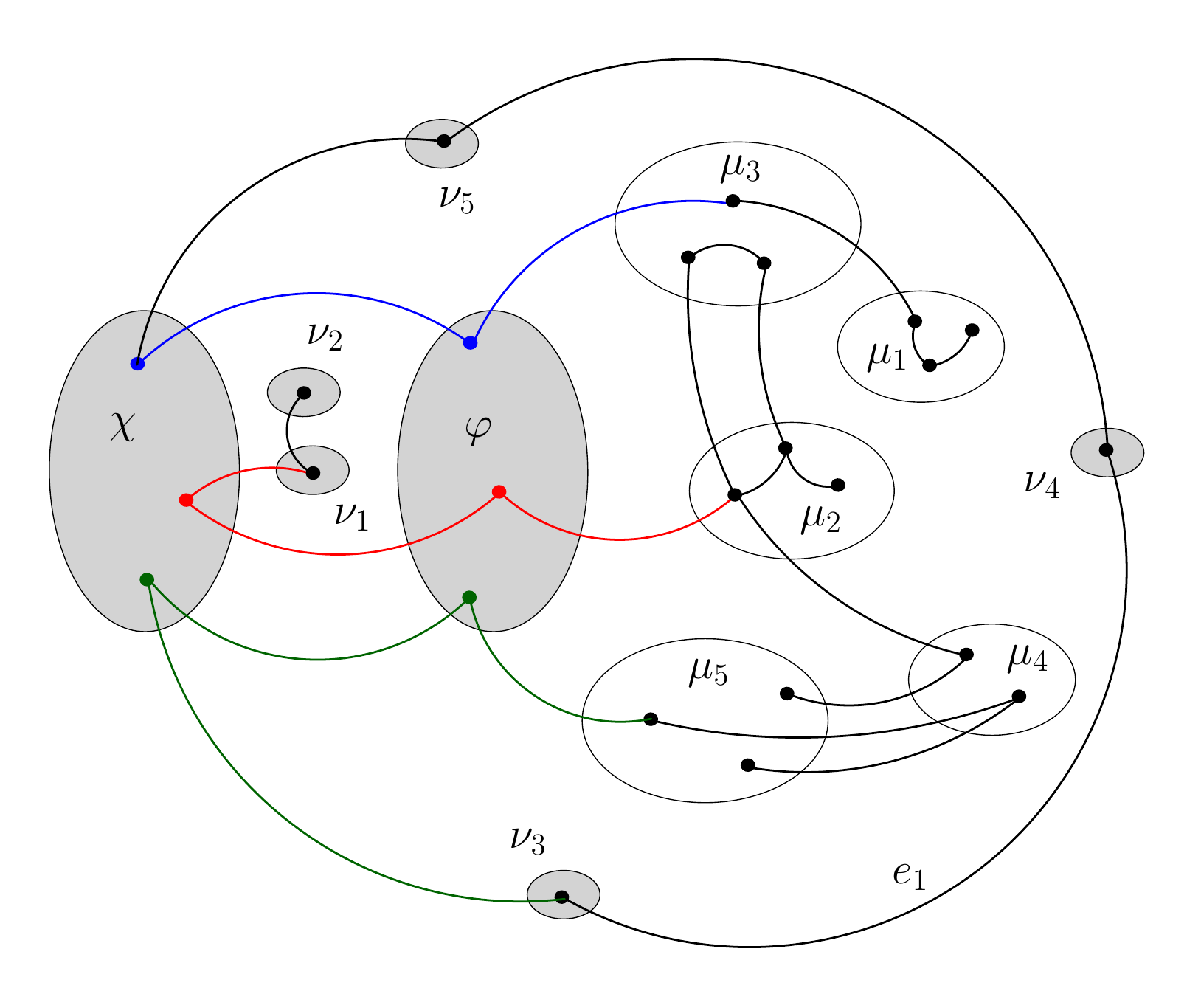}\label{fi:ifc2fc-a}}
\hfil
\subfigure[]{\includegraphics[page=2,height=0.35\textwidth]{figures/ifc2fc.pdf}\label{fi:ifc2fc-b}}
\caption{\subref{fi:ifc2fc-a} A c-planar drawing of the flat c-graph of Fig.~\ref{fi:construction2-b}. \subref{fi:ifc2fc-b} The corresponding c-planar drawing the flat c-graph of Fig.~\ref{fi:construction2-a}. The gray region is $R(\mu^*)$.}\label{fi:ifc2fc}
\end{figure}

Converserly, given a c-planar drawing $\Gamma(C_{i+1})$ of the flat c-graph $C_{i+1}$, we show how to construct a c-planar drawing $\Gamma(C_{i})$ of $C_{i}$ (refer to Fig.~\ref{fi:ifc2fc}). 
Consider the regions corresponding to the independent clusters $\nu_i$, with $i=1,\dots,h$, containing the nodes that were originally children of $\mu^*$. These regions (filled white in Fig.~\ref{fi:ifc2fc-a}) may be without any inter-cluster edge; may have inter-cluster edges among themselves; and may be connected to $R(\chi)$. However, by construction these regions cannot have inter-cluster edges connecting them to $R(\varphi)$, or connecting them to the regions of the original children $\mu_i$ of $\rho$. Hence, the regions corresponding to $\nu_1$, \dots, $\nu_h$ can be classified into two sets, denoted $\mathcal{A}_\chi$ and $\mathcal{F}_\chi$, of `anchored regions' and `floating regions' of $\chi$, respectively, where an \emph{anchored region of $\chi$} is a region $R(\nu_a)$ containing a vertex of $G_{i+1}$ that is connected (via a path) to a vertex in $\chi$ and a \emph{floating region of $\chi$} is a region $R(\nu_f)$ containing a vertex that is not connected to vertices in~$\chi$. 

Analogously, consider the regions $R(\mu_j)$, with $j=1, \dots, l$, corresponding to the original children $\mu_j \neq \mu^*$ of $\rho$ (filled gray in Fig.~\ref{fi:ifc2fc-a}). These regions may be without any inter-cluster edge; may have inter-cluster edges among them; or may have inter-cluster edges connecting them to $R(\varphi)$. However, by construction these regions cannot have inter-cluster edges connecting them to $R(\chi)$, or connecting them to the the regions in $\mathcal{F}_\chi$ or $\mathcal{A}_\chi$. Hence, we can classify the regions corresponding to $\mu_1$, \dots, $\mu_l$ into two sets, denoted $\mathcal{A}_\varphi$ and $\mathcal{F}_\varphi$, of `anchored regions' and `floating regions' of $\varphi$, where an \emph{anchored region of $\varphi$} is a region $R(\mu_a)$ whose cluster $\mu_a$ contains at least one vertex of $G_{i+1}$ that is connected to a vertex in $\varphi$ and a \emph{floating region of $\varphi$} is a region $R(\mu_f)$ whose cluster $\mu_f$ contains all vertices not connected to vertices in $\varphi$. 

Our strategy will be that of removing altogether from $\Gamma(C_{i+1})$ the drawings of the floating regions (and all their content), possibly modifying the drawing of the remaining graph, and then suitably reinserting the drawing of the floating regions. 

Suppose now to have temporarily removed from $\Gamma(C_{i+1})$ the drawings of the floating regions in $\mathcal{F}_\chi$ and $\mathcal{F}_\varphi$. 
We define an auxiliary multigraph $H$ that has one vertex $v_\chi$ representing $\chi$ and one vertex $v_{\nu_i}$ for each singleton $\nu_i$ introduced when removing $\mu^*$ such that $R(\nu_i) \in \mathcal{A}_\chi$. For each inter-cluster edge between two clusters $\lambda_1$ and $\lambda_2$ corresponding to the vertices $v_{\lambda_1}$ and $v_{\lambda_2}$ of $H_\chi$, respectively, we add an edge $(v_{\lambda_1},v_{\lambda_2})$ to $H$. Observe that $H$, by the definition of the anchored regions in $\mathcal{A}_\chi$, is connected.

Drawing $\Gamma(C_{i+1})$ induces a drawing $\Gamma(H)$ of the multigraph $H$, where each vertex $v_\lambda$ of $H$ is represented by the region $R(\lambda)$ of the cluster $\lambda$ corresponding to $v_\lambda$ and each edge $(v_{\lambda_1},v_{\lambda_2})$ of $H$ is represented as the corresponding inter-cluster edge of $\lambda_1$ and $\lambda_2$ restricted to the portion that is drawn outside the boundaries of $R(\lambda_1)$ and $R(\lambda_2)$. 

Two are the cases: either $\Gamma(H)$ does not contain in one of its internal faces $R(\varphi)$ (\textsc{Case 1}) or it contains $R(\varphi)$ (\textsc{Case 2}, depicted in Fig.~\ref{fi:ifc2fc-a}). 

In \textsc{Case 1} no change has to be done to $\Gamma(C_{i+1})$. In \textsc{Case 2} we modify $\Gamma(H)$ and, consequently, $\Gamma(C_{i+1})$ so to fall again into \textsc{Case 1}. Namely, we identify a minimal set $\{e_1, e_2, \dots, e_q\}$ of edges of $H$ that, if removed, would bring $R(\varphi)$ on the external face of $\Gamma(H)$ (for example in Fig~\ref{fi:ifc2fc-a} this set contains only edge $e_1$). Starting from edge $e_1$, that is incident to the external face of $\Gamma(H)$, we redraw each $e_i$, with $i = 1, \dots, q$, as follows. Suppose that the curve for $e_i=(v_{\lambda_1},v_{\lambda_2})$ in $\Gamma(H)$ starts from a point $p_1$ on the boundary of $R(\lambda_1)$ and ends with a point $p_2$ on the boundary of $R(\lambda_2)$. We arbitrarily choose two distinct points $p_3$ and $p_4$, encountered in this order when traversing $e_i$ from $p_1$ to $p_2$. We remove the portion of $e_i$ between $p_3$ and $p_4$ and we redraw it by returning back from $p_3$ towards $p_1$ on the external face of $\Gamma(H)$ and then moving along the external face of $\Gamma(H)$ until we reach $p_4$ (see, for example, Fig.~\ref{fi:ifc2fc-b}). Observe that this corresponds to moving the external face of $\Gamma(H)$ to a face that was previously an internal face of $\Gamma(H)$ enclosed by $e_i$. We carry on doing the same operation for each $e_i$, with $i=1, \dots, q$, until the external face of $\Gamma(H)$ is incident on the boundary of $R(\varphi)$. At this point we are in \textsc{Case 1}.    

Observe that since $\Gamma(H)$ does not contain in one of its internal faces $R(\varphi)$, then it cannot contain any region in $\mathcal{A}_\varphi$ either, as, by definition, these regions are connected to $R(\varphi)$. Hence, the internal faces of $\Gamma(H)$ only contain vertices and edges that in $C_i$ belong to $\mu^*$. 

Now we reinsert the drawings of the floating regions. We identify an arbitrarily small empty disk $F_\chi$ inside $R(\chi)$ and move inside $F_\chi$ the (suitably scaled down) drawings of the floating regions in $\mathcal{F}_\chi$. Analogously, we identify an arbitrarily small empty disk $F_\varphi$ inside $R(\varphi)$ and move inside $F_\varphi$ the (suitably scaled down) drawings of the floating regions in $\mathcal{F}_\varphi$. Consider the region $R(\mu^*)$ that is the region covered by $\Gamma{H}$. Such a region is connected, is simple, contains only vertices and nodes of $\mu^*$, and its boundary is a simple curve (see Fig.~\ref{fi:ifc2fc-b}). Therefore, by neglecting the boundaries of $R(\chi)$, $R(\varphi)$, $\nu_1$, $\nu_2$, \dots, $\nu_h$ and by removing the internal vertices of $R(\chi)$ and $R(\varphi)$ and joining their incident edges we obtain a c-planar drawing $\Gamma(C_i)$ of~$C_i$. 
\end{proof}

\section{-- Proof of Lemmas~\ref{le:properties2} and~\ref{le:properties3} and of Observation~\ref{ob:achieved} of Section~\ref{se:discussion2}.}\label{ap:discussion2}

\bigskip
\noindent
\textbf{Lemma~\ref{le:properties2}.} \emph{
Let $C_f(G_f,T_f)$ be an $n_f$-vertex flat clustered graph with $c_f$ clusters. The independent flat clustered graph $C_{\textrm{if}}(G_{\textrm{if}},T_{\textrm{if}})$ equivalent to $C_f$ built as described in the proof of Theorem~\ref{th:main2} has the following properties:
\begin{enumerate}
\item\label{pr:subdivision2} Graph $G_{\textrm{if}}$ is a subdivision of $G_f$
\item\label{pr:edges2} Each inter-cluster edge of $G_f$ is replaced by a path of length at most $4$.
\item\label{pr:vertices2} The number of vertices of $G_{\textrm{if}}$ is $O(n_f)$ 
\item\label{pr:clusters2} The number of clusters of $C_{\textrm{if}}$ (including the root) is $c_{\textrm{if}} \leq 2c_f + n_f -1$
\end{enumerate}
}

\begin{proof}
Property~\ref{pr:subdivision2} descends from the fact that each step of the transformation of $G_{\textrm{f}}$ into $G_{\textrm{if}}$ consists of edge subdivisions only. In particular, every inter-cluster edge of $\mu^*$ is subdivided twice when removing the non-independent cluster $\mu^*$. It follows that inter-cluster edges are replaced by paths of length at most $4$ (exactly $4$ if the edge links two non-independent clusters). This proves Property~\ref{pr:edges2}. Since by Property~\ref{pr:edges2} each edge is replaced by a path of bounded length and $G_\textrm{f}$ has $O(n_\textrm{f})$ edges, the number of vertices of $G_{\textrm{if}}$ is $O(n_\textrm{f})$ (Property~\ref{pr:vertices2}).
In order to prove Property~\ref{pr:clusters2} observe that when removing a non-independent cluster $\mu^*$ two new clusters are introduced and all vertices of $\mu^*$ are enclosed into new singleton clusters. Therefore, the number $c_{\textrm{if}}$ of clusters of $C_{\textrm{if}}$ is at most $2c_\textrm{f} + n_\textrm{f} -1$ (the minus 1 is due to the fact that the root cluster does not need to be removed).
\end{proof}

\bigskip
\noindent
\textbf{Observation~\ref{ob:achieved}.} \emph{
At the same asymptotic cost of the reduction described in the proof of Theorem~\ref{th:main2} it can be achieved that non-root clusters are of two types: (\textsc{Type~1}) clusters containing a single vertex of arbitrary degree or (\textsc{Type~2}) clusters containing multiple vertices of degree two. 
}

\begin{proof}
The property is achieved if, in addition to removing non-independent clusters of the instace $C_\textrm{f}(G_\textrm{f},T_\textrm{f})$, we also use the same technique described in the proof of Theorem~\ref{th:main2} to remove those independent clusters of $C_\textrm{f}$ that contain at least one vertex of degree greater than $2$. In this case all clusters of $C_{\textrm{if}}$ that contain more than one vertex are guaranteed to have all degree-two vertices. The cost of the reduction is still linear for the same reasons discussed in the proof of Theorem~\ref{th:main2}.
\end{proof}

\bigskip
\noindent
\textbf{Lemma~\ref{le:properties3}.} \emph{
Let $C(G,T)$ be an $n$-vertex clustered graph with $c$ clusters. The independent flat clustered graph $C_{\textrm{if}}(G_{\textrm{if}},T_{\textrm{if}})$ equivalent to $C$ built by concatenating the reduction of Theorem~\ref{th:main} and the reduction of Theorem~\ref{th:main2}, as modified by Observation~\ref{ob:achieved}, has the following properties:
\begin{enumerate}
\item\label{pr:subdivision3} Graph $G_{\textrm{if}}$ is a subdivision of $G$
\item\label{pr:edges3} Each inter-cluster edge of $G_f$ is replaced by a path of length at most $4h(T)-4$
\item\label{pr:vertices3} The number of vertices of $G_{\textrm{if}}$ is $O(n^2)$ 
\item\label{pr:clusters3} The number of clusters of $C_{\textrm{if}}$ is $O(n \cdot h(T))$
\item\label{pr:cluster-types3} Non-root clusters are of two types: (\textsc{Type~1}) clusters containing a single vertex of arbitrary degree or (\textsc{Type~2}) clusters containing multiple vertices of degree two 
\end{enumerate}
}

\begin{proof}
Properties~\ref{pr:subdivision3}, \ref{pr:vertices3}, and \ref{pr:clusters3} directly descends by concatenating the analogous Properties~\ref{pr:subdivision}, \ref{pr:vertices3}, and \ref{pr:clusters3} of Lemmas~\ref{le:properties} and~\ref{le:properties2}. Property~\ref{pr:cluster-types3} is a direct consequence of Observation~\ref{ob:achieved}. The only property that needs a detailed proof is Property~\ref{pr:edges3}. 
By Property~\ref{pr:edges} of Lemma~\ref{le:properties} the first transformation of $C(G,T)$ into the flat c-graph $C_\textrm{f}(G_\textrm{f},T_\textrm{f})$ replaces an edge with a path of length at most $4h(T)-8$ (see Figs.~\ref{fi:edge-subdivision-a} and~\ref{fi:edge-subdivision-b}). When transforming $C_\textrm{f}(G_\textrm{f},T_\textrm{f})$ into the independent flat c-graph $C_\textrm{if}(G_\textrm{if},T_\textrm{if})$ only the original lower clusters of $C$ ($\mu_1$ and $\mu_2$ in the example of Fig.~\ref{fi:edge-subdivision}) need to be replaced, since the cluster introduced by the first transformation are already independent and of \textsc{Type 2}. By Property~\ref{pr:edges2} of Lemma~\ref{le:properties2} this adds $4$ more internal vertices to each replaced edge (see Fig.~\ref{fi:edge-subdivision-c}). Hence, each edge of $G$ is replaced by a path of length at most $4h(T)-4$ in $G_\textrm{if}$.   
\end{proof}

\begin{figure}[tb]
\centering
\subfigure[]{\includegraphics[page=1,width=0.50\textwidth]{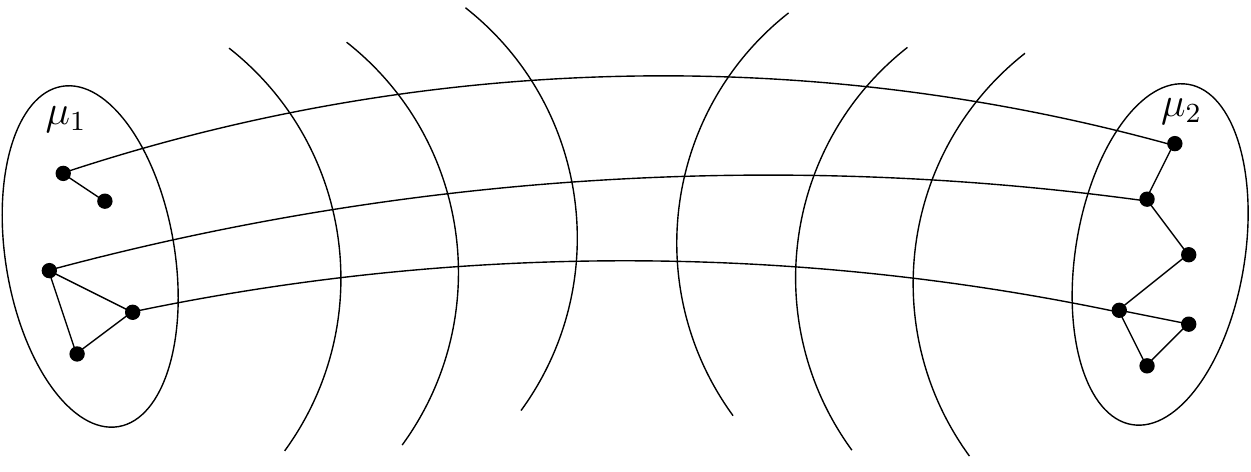}\label{fi:edge-subdivision-a}}
\hfil
\subfigure[]{\includegraphics[page=2,width=0.50\textwidth]{figures/edge-subdivision.pdf}\label{fi:edge-subdivision-b}}
\hfil
\subfigure[]{\includegraphics[page=3,width=0.50\textwidth]{figures/edge-subdivision.pdf}\label{fi:edge-subdivision-c}}
\caption{A figure for the proof of Lemma~\ref{le:properties3}. \subref{fi:edge-subdivision-a} An example of a c-graph where $h(T)=5$. Edges connecting vertices in two lower clusters $\mu_1$ and $\mu_2$ traverse at most $2h(T)-4=6$ boundaries of higher clusters. \subref{fi:edge-subdivision-b} The corresponding flat c-graph obtained as described in the proof of Theorem~\ref{th:main} replaces each edge with a path of at most $4h(T)-8=12$. \subref{fi:edge-subdivision-c} The final independent flat c-graph obtained as described in the proof of Theorem~\ref{th:main2} replaces each original edge with a path of length at most $4h(T)-4=16$.}\label{fi:edge-subdivision}
\end{figure}

\fi

\end{document}